\documentclass[draft]{article}
\usepackage{amsmath,amsfonts,latexsym, amssymb}

\usepackage{color}

\newcommand{\Tr}{\mbox{{\rm Tr}\,}}

\newcommand{\e}{\varepsilon}

\newcommand{\rd}{{\rm d}}
\newcommand{\bR}{{\mathbb R}}
\newcommand{\bbZ}{{\mathbb Z}}

\newcommand{\bs}{ \begin{split} }
\newcommand{\es}{ \end{split} }

\newcommand{\bE}{{\bf{E}}}

\newcommand{\bN}{{\mathbb N}}

\newcommand{\al}{\alpha}

\newcommand{\be}{\begin{equation}}
\newcommand{\ee}{\end{equation}}

\newcommand{\om}{{\omega}}

\newcommand{\cE}{{\cal E}}

\newcommand{\C}{{\mathbb C}}
\newcommand{\Z}{{\mathbb Z}}
\newcommand{\E}{{\bf{E}}}
\newcommand{\R}{{\mathbb R}}
\newcommand{\N}{{\mathbb N}}
\renewcommand{\P}{{\bf{P}}}

\newcommand{\wt}{\widetilde}

\newtheorem{theorem}{Theorem}
\newtheorem{corollary}[theorem]{Corollary}
\newtheorem{lemma}[theorem]{Lemma}
\newtheorem{proposition}[theorem]{Proposition}

\newtheorem{definition}{Definition}
\newcommand{\qed}{\hfill\fbox{}\par\vspace{0.3mm}}
\newenvironment{proof}{{\bf Proof.}} {\hfill\qed}

\numberwithin{equation}{section}
\numberwithin{theorem}{section}
\numberwithin{definition}{section}

\title{Wegner estimate and Anderson localization for random magnetic
fields}
\author{L\'aszl\'o Erd\H os\thanks{Partially supported by SFB-TR12 of
the German Science Foundation},\; David Hasler
 \\
\\
Institute of Mathematics, University of Munich, \\
Theresienstr. 39, D-80333 Munich, Germany \\
\text{lerdos@math.lmu.de, hasler@math.lmu.de} }

\date{Dec 23, 2010}

\begin{document}

\maketitle

\begin{abstract}
We consider a two dimensional magnetic Schr\"odinger operator
with a spatially stationary random magnetic field.
We assume that the magnetic field has a positive lower bound
and that it has Fourier modes on arbitrarily short scales. We prove
the Wegner estimate at arbitrary energy, i.e.
we show that the averaged density of states is finite
throughout the whole spectrum. We also prove Anderson localization
at the bottom of the spectrum.
\end{abstract}


\medskip

{\it Running title:} Wegner estimate for magnetic field

\medskip



\section{Introduction}

We consider a spinless quantum particle in the two dimensional
Euclidean space $\bR^2$
  subject to a random magnetic field $B:\bR^2\to\bR$. The energy is given
by the
magnetic Schr\"odinger operator, $H=(p-A)^2+V$, where $p=-i\nabla$,
$A:\bR^2\to\bR^2$ is a random
magnetic vector potential satisfying $\nabla\times A=B$
and $V$ is a deterministic external potential.
In contrast to the standard Anderson model for localization
with a magnetic field (see, e.g., \cite{CH, DMP, FLM, GK, W}), we
consider a model where the external potential is deterministic,
and only the magnetic field carries   randomness in the system.

The existence of the integrated density of states and its
independence of the boundary conditions in the thermodynamic
limit have been proven for both the discrete and the
continuous model and Lifschitz tail  asymptotics
have also been obtained \cite{HS, Na1, Na2}.

However,  Anderson localization for the random field model
has only been shown under an additional condition that the
random part of the magnetic flux is locally zero \cite{KNNN}.
Since a deterministic constant magnetic field localizes, one could 
expect that its random perturbation even enhances localization, so the
zero flux condition in  \cite{KNNN} should physically be unnecessary.
Technically, however, random magnetic fields are harder to fit
into the standard proofs of localization mainly because
the vector potential is nonlocal while a spatially stationary 
magnetic field typically does not lead
to a stationary Hamiltonian.

To circumvent this difficulty,
Hislop and Klopp \cite{HK} and later Ueki \cite{U}
have considered  spatially stationary random vector potential of the form
\be
   A_\om(x) = \sum_{z\in\bbZ^2}\om_z u(x-z),
\label{Akl}
\ee
  where $\om=\{\omega_z \; : \;
z\in \bbZ^2\}$ is a collection of i.i.d. real random variables with
some moment condition
and $u:\bR^2\to\bR^2$ is a fixed vectorfield with a fast decay at
infinity.
For such random field, Anderson localization was shown
in \cite{U,GHK}, motivated by a method in \cite{HK}, that gave
the first Wegner estimate  for this model. The method works
  only for  energies away from the spectrum
of the deterministic part of  the Hamiltonian,
mainly because the Wegner estimate is shown only in that regime.

Note that
for the magnetic field $B_\om=\nabla\times A_\om$ generated by
\eqref{Akl},  the fluctuation of the total flux $\int_{\Lambda_L} B_\om(x)\rd x$
within a large box  $\Lambda_L= [-L/2, L/2]^2$
is of order of the square root of the boundary, $|\partial \Lambda_L|^{1/2}\sim
L^{1/2}$ by the central limit theorem.
  In contrast, if the magnetic
field $B_\om(x)$ itself were given by a spatially stationary
random process
with a sufficient correlation decay, e.g.,
\be
B_\om(x) = \sum_{z\in\bbZ^2}\om_z u(x-z)
\label{rmf}
\ee
with some decaying scalar function $u:\bR^2\to \bR$,
then  $\int_{\Lambda_L} B_\om(x)\rd x$ would fluctuate on a scale of order
square root of the area, $|\Lambda_L|^{1/2}\sim L$.
  Assuming stationarity
on the vector potential thus imposes
an unnatural constraint on the physically relevant
gauge-invariant quantity, i.e. on the magnetic field.

The analogous problem  for the lattice magnetic Schr\"odinger operator
has been studied with different methods. For the discrete
magnetic Schr\"odinger operator on $\bbZ^2$, the magnetic field
is given by its flux on each plaquet of the lattice.
  Extending the method of Nakamura \cite{Na1}, Anderson localization
was proven for this model \cite{KNNN} near
the spectral  edge, however, the
 zero flux condition was enforced in a strong sense.
Instead of considering the more natural
i.i.d. (or weakly correlated) random fluxes
on each plaquet, the neighboring plaquets were domino-like paired
and the magnetic fluxes were opposite
within each domino. Such magnetic field again has much less
fluctuation than the i.i.d. case, moreover the flux is
deterministically zero on each domino.

The main technical reason for
the zero flux condition in both the continuous and the
discrete model was that the proof of the Wegner estimate
required it. The Wegner estimate is a key element in any known
mathematical proof of the Anderson localization since
it provides an a-priori bound for the resolvent
with a very high probability.
Typically, the statement is formulated
  for the finite volume truncation
$H_L$ of $H$ onto a box $\Lambda_L=[-L/2,
L/2]^2$
with some boundary conditions. The Wegner estimate states that
the expected number of eigenvalues of $H_L$ within a small
spectral interval $I$ is bounded from above by $C(L)|I|$,
where $|I|\ll 1$ denotes the length
of the interval.  As $|I|\to 0$, this provides an upper bound on
the averaged density of states and Lipschitz continuity of
the averaged integrated density of states.
Ideally, the constant $C(L)$ should be proportional
with the volume of the box, but for the purpose of
Anderson localization $C(L)$ often  may even grow subexponentially with
the volume. Moreover, it is also sufficient if the averaged
integrated density of states is only H\"older continuous, which
corresponds to a bound $C(L)|I|^\alpha$, $0<\alpha<1$, for
the expected number of eigenvalues in the Wegner estimate.

In this paper we present a new method to prove a Wegner estimate
that applies to a certain class of spatially stationary random
magnetic fields and to any energy in the spectrum.
Our estimate gives the optimal (first) power of $|I|$,
but not the optimal volume dependence:  $C(L)$ is
a high (but universal)  power of $L$.
As an application of the Wegner estimate, we prove Anderson localization for our
model at the bottom of the spectrum. We will address
the localization at higher energies  later.

We remark that our new approach can also be used to prove a Wegner estimate
and localization for the discrete Schr\"odinger operator
with a random magnetic field given by i.i.d. random fluxes
on the plaquets of $\Z^2$. The details will be given in a 
separate paper \cite{EH2}.


\section{Definition of the random magnetic field}\label{sec:field}

We work in $\bR^2$ and we set $|x|_\infty := \max\{ |x_1|, |x_2|\}$
for any $x\in \bR^2$.  We are given two
positive numbers, $b_0$ and $K_0> 3$, and
a deterministic (possibly nonconstant) magnetic field
$B_{\rm det}(x)$ with
\be
    0< 2 b_0\leq  B_{\rm det}(x) \leq (K_0 - 1 ) b_0.
\label{detbound}
\ee
We perturb this magnetic field by a random one, i.e., we consider 
\be
B =     B_\omega= B_{\rm det}+ \mu B_{\rm ran}^\omega,
\label{B}
\ee
where the random field is assumed to be $|B_{\rm ran}^\omega|\leq b_0$ and
$0<\mu\le 1$ is a coupling constant. In particular,
\be
    0< b_0 \leq B_\omega  \leq K_0 b_0.
\label{lowbound}
\ee

\smallskip

Now we define the random field $B_{\rm ran}^\omega$ more precisely.
We will need the assumption that  $B_{\rm ran}^\omega$
has  components on
arbitrary small scales, but these components   decay in
size. For simplicity we present a class of magnetic
fields for which our method works, but our approach
can be extended to more general fields with
a similar structure.
We remark that the analogous result in the discrete setup
 \cite{EH2}
will not require such assumption on the structure
of the random field.

We choose a smooth profile function $u\in C_0^1(\bR^2)$, $0\le u\leq 1$,
that satisfies one of the following two conditions for some
sufficiently small $\delta$:
either
\be
   u(x)\equiv 0  \quad \mbox{for}
   \quad |x|_\infty  \geq \frac{1}{2}+\delta
    \quad\mbox{and}\quad     u(x) \equiv 1, \quad \mbox{for}
   \quad |x|_\infty  \leq
   \frac{1}{2} - \delta
\label{kcond}
\ee
or
\be
    u(x) = \delta^2u_0( x\delta)
  \;\;\; \mbox{with some $u_0\in C^1_0(\bR^2)$},\;\; \int_{\bR^2} u_0=1,
\quad
u_0(x)\equiv0\;\;
\mbox{for $|x|_\infty \ge 1$}.
\label{kcond1}
\ee
In both cases $\delta$ will be chosen as a sufficiently small
positive number $\delta\leq\delta_0\leq 1$. The threshold $\delta_0$
can be chosen as
\be
\begin{split}
  \delta_0 =&\frac{1}{3200} \quad \mbox{under condition \eqref{kcond}}
\label{delta0}
\cr
\delta_0 = & \frac{1}{640+32\|\nabla u_0\|_\infty^2}
    \quad \mbox{under condition \eqref{kcond1}}.
\end{split}
\ee

Fix $k\in \bN$  and define the lattice $\Lambda^{(k)} = (2^{-k}\bbZ)^2$.
For $z\in \Lambda^{(k)}$ define
\be
       \beta_z^{(k)}(x) :=  u\big(2^{k}(x-z)\big).
\label{betadef}
\ee
The randomness is represented by a  collection of
independent random variables 
$$
\omega = \{ \om^{(k)}_z\; : \; k\in\bN,
z\in \Lambda^{(k)}\}.
$$ 
We assume that all $\omega^{(k)}_z$ have zero
expectation, and they satisfy a bound that is uniform in $z$
\be
     | \om^{(k)}_z| \leq \sigma^{(k)}:= e^{-\rho k}
\label{lowbound1}
\ee
with some $\rho>0$. We assume that the distribution of $\omega_z^{(k)}$ is absolutely
continuous, its density function  $v^{(k)}_z$ is in $C^2_0(\bR)$
and  satisfies
\be
    \int_{\bR} \Big| \frac{\rd^2 v^{(k)}_z}{\rd s^2}(s)\Big|
  \rd s \leq C[\sigma^{(k)}]^{-2} = Ce^{2\rho k}.
\label{2der}
\ee
Note that we do not require identical distribution.
Thus for each  $(k,z)  \in \mathcal{L} :=
\bigcup_{k \in \N} \{ k \} \times \Lambda^{(k)} $ we have a
 probability measure with
density $v_z^{(k)}$. The associated product measure, $\P$,
 is probability measure on
$\Omega = \R^{\mathcal{L}}$, and we denote expectation 
with respect to this measure
by $\E$.

For example,  one can assume that for each fixed $k$, the
random variables $\{ \om^{(k)}_z\; : \; z\in \Lambda^{(k)}\}$
are i.i.d. and they all live on a scale $\sigma^{(k)}$,
e.g.  $v^{(k)}_z(s) = [\sigma^{(k)}]^{-1}v(s/\sigma^{(k)})$
for all $z$
with some smooth, compactly supported density function $v$.

We define the random magnetic field as
\be
 B_{\rm ran}^\omega(x) =  B_{\rm ran}(x) := \sum_{k=0}^\infty  B^{(k)}(x), \qquad
B^{(k)}(x) := \sum_{z\in \Lambda^{(k)}} B^{(k)}_z(x), \qquad
     B^{(k)}_z(x) := \om^{(k)}_z   \beta_z^{(k)}(x) \; ,
\label{bconst}
\ee
i.e. $B_{\rm ran}^\omega$ is the sum of independent local magnetic fields on each scale 
$k$ and at every $z \in \Lambda^{(k)}$. 
We assume that
\be \label{newcondranb}
    \sum_{k=0}^\infty \sigma^{(k)}\leq b_0,
\ee
i.e.
$(1-e^{-\rho}) b_0\ge 1$,
then clearly $|B_{\rm ran}(x)|\leq b_0$.
Note that $B^\omega_{\rm ran}$  is differentiable if $\rho > \ln 2$.
We will make the following assumption

\medskip

\noindent
{\bf (R)}   $B_{\omega}$ is a  random
magnetic field  constructed   in \eqref{B}, \eqref{betadef}, and \eqref{bconst}, and it 
satisfies  \eqref{detbound},  \eqref{lowbound}, \eqref{lowbound1}, \eqref{2der}, 
 \eqref{newcondranb}, and one of the conditions  \eqref{kcond} or \eqref{kcond1}.

\medskip

Let $\Lambda \subset \bR^2$ be square    and we will consider the
magnetic Sch\"odinger operator  with Dirichlet boundary conditions
on $\Lambda$.
We will work in the Hilbert space $L^2(\Lambda)$ and
  denote the scalar product by $\langle \,\cdot ,\cdot\, \rangle$
and the norm by $\|\,\cdot\,\|$.
Let $A$ be a magnetic vector potential such that $\nabla \times A = B$.
By $H_\Lambda(A)$  we 
denote the magnetic Schr\"odinger operator on $L^2(\Lambda)$  with Dirichlet boundary
conditions, i.e.,
$H_\Lambda(A) = (p - A)^2  + V $.
Here $V$ is a bounded external potential.
In the special case where  $\Lambda_L = [-L/2, L/2]^2 \subset \R^2$ with 
 $L\in \bN$  we will write 
\begin{equation} \label{eq:defofdirichletham}
H_L(A)  = H_{\Lambda_L}(A) .
\end{equation}
By $H(A) = (p - A)^2 + V $, we denote the magnetic Schr\"odinger operator on $L^2(\R^2)$.
The magnetic Hamilton operators  can be realized by  the Friedrichs extension.
If we refer to statements which are independent of the particular choice of gauge,
with a slight abuse of notation, we shall
occasionally write $H_\Lambda(B)$ and $H(B)$.
If $\nabla \times A_\omega = B_\omega$ and $B_\omega$ satisfies {\bf (R)}, then 
 $\omega \to H_L(A_\omega)$ is measurable. This follows for example from an application of 
Proposition 1.2.6 \cite{S}.

\section{Main results} \label{sec:thm}

The first result is a Wegner estimate.
Fix an energy $E$ and a window of width $\eta\leq 1$ about $E$.
Let $\chi_{E,\eta}$ be the characteristic function of the interval
$[E-\eta/2, E + \eta/2]$.

\begin{theorem}\label{thm:wegner} Let $K_0 > 3$.
We assume that $B_\omega$ is a random magnetic field satisfying {\bf (R)}.
Let $A_\omega$ be a vector potential
with $\nabla\times A_\omega =B_\omega = B_{\rm det} + \mu B_{\rm ran}^\omega$ with $\mu\in (0,1]$.
  We assume $\rho \ge\ln 2$,
  $(1-e^{-\rho})b_0\ge 1$, and $\| V \|_\infty \leq b_0/4$.
Let $\delta\leq \delta_0$ and 
$K_1 \geq 1$. Then 
there  exist positive  constants $C_0=C_0(K_0,K_1)$, $C_1 = C_1(K_0,K_1)$, 
and $L_0^*=L_0^*(K_0,K_1,\delta)$  such that
for any $0 < \kappa \leq 1$ 
$$
   \bE \, \Tr \chi_{E,\eta}(H_L(A)) \leq C_0  \eta \mu^{-2}
L^{ C_1(\kappa^{-1}  + \rho)} , 
$$
for all  $E \in [ \frac{b_0}{2} ,  K_1 b_0]  $,  $0 < \eta \leq 1$, and  $L\ge L_0^* b_0^{\kappa}$.
\end{theorem}
The next theorem is a standard result stating that the spectrum is deterministic.
For this we need that the random magnetic field is stationary  on each scale:

\medskip
\noindent
{\bf (i.i.d.)}  For any fixed $k \in \N$,  $\{ \omega_z^{(k)} : z \in \Lambda^{(k)} \}$  
are i.i.d., i.e.,  { $v^{(k)} = v^{(k)}_z$} .
\medskip

\begin{theorem} \label{thm:deterministic} Suppose $B_\omega$ is a random magnetic field such that
{\bf (R)} and {\bf  (i.i.d.)} hold. For  $\nabla \times A_\omega = B_\omega$, the function 
 $\omega \mapsto H(A_\omega)$ is measurable.
There exists a 
set $\Sigma \subset \R$  and  a set
 $\Omega_1 \subset \Omega$ with  ${\bf P}(\Omega_1) = 1$ such that for all $\omega \in \Omega_1$
$$
\sigma( H(B_\omega)) = \Sigma.
$$
\end{theorem}
 For completeness, we give a  proof of Theorem \ref{thm:deterministic} in  Appendix A.

The second main result is localization at the bottom of the spectrum.
We make additional assumptions on the profile function, namely that
\begin{equation} \label{eq:profileidentity1}
  U(x): =  \sum_{z \in \Z^2} u(x - z) \geq c_u  \qquad \mbox{and}        
 \quad  {\rm sup}_{x \in \R^2} U(x)   = 1 ,
\end{equation}
for some  positive constant $c_u >0$.

The result about localization will hold  under  the  following Hypotheses.

\medskip

\noindent
{\bf (A)}    $V$   and $B_{\rm det}$ are  $\Z^2$-periodic, and  
 $B_\omega$ is a  random
magnetic field  satisfying {\bf (R)}  with   profile function
satisfying
\eqref{eq:profileidentity1}.   Hypothesis {\bf (i.i.d.)} holds and ${\rm supp} v^{(k)}$ is 
a compact interval $[m_-^{(k)},m_+^{(k)}]$.

\medskip

\medskip

Second we assume a polynomial bound on the lower tail of $v^{(0)}$. To this end we 
introduce the probability distribution function 
\be \label{eq:probdist1}
 \nu(h) := \int_{m_-^{(0)}}^{m_-^{(0)}+ h } v^{(0)}(x) \rd x  .
\ee
\medskip

\noindent
 {\bf ($\boldsymbol{{\rm A}_\tau}$)} Hypotheses 
{\bf (A)} holds and  there exists a constant $c_v$ 
 such that for all $h \geq 0$ we have 
$
 \nu(h) \leq c_v h^\tau$.

\medskip

The next theorem states that we have Anderson localization at the bottom of the spectrum.
Recall that $\Sigma$  denotes the almost sure
 deterministic spectrum of  $H(B_\omega)$, see Theorem  \ref{thm:deterministic}, and let 
  $\Sigma_{\rm inf}$ be its infimum. 
We will  assume that  the following 
quantity is finite 
\begin{align*}
& K_2 :=   \max_{|\alpha|=1} \left[ (1- 2 e^{- \rho})^{-1}   \| D^\alpha  U  \|_\infty  + 
  \| D^\alpha  B_{\rm det} \|_\infty  +  \| D^\alpha V \|_\infty    \right]  .
\end{align*}
where  we used the multi-indices notation
$D^\alpha = \partial_{x_1}^{\alpha_1}\partial_{x_2}^{\alpha_2}$ with $\alpha \in \N_0^2$
and $|\alpha|= \alpha_1 + \alpha_2$. To show localization we will need that 
$K_2 b_0^{-1/2}$ is small. A more  explicit relation between  $b_0$ and  derivatives 
of $U$, $B_{\rm det}$, and $V$ 
can be  obtained from the first inequality of  \eqref{eq:bdepend}  given  in the proof.

\begin{theorem} \label{thm:loclization} 
Let    $K_0 > 3$, $K_1 \geq 1$, 
$\rho > \ln 2$,
  $(1-e^{-\rho})b_0\ge 1$,  $\| V \|_\infty \leq b_0/4$.
Suppose  {\bf ($\boldsymbol{{\rm A}_\tau}$)} holds for some  $\tau > 2 $,
and let $B_\omega = B_{\rm det}+ \mu B_{\rm ran}^\omega$ with $\mu\in (0,1]$ be the random
magnetic field with a vector potential  $A = A_\omega$. 
If  $K_2 b_0^{-1/2}$  is sufficiently small, then 
there exists an $\e_0>0$ such that for almost every $\omega$ the operator  $H(A_\omega)$ has in 
 $[\Sigma_{\rm inf}, \Sigma_{\rm inf}  + \e_0]$
 dense  pure point spectrum with exponentially decaying eigenfunctions.
For  $p < 2  (\tau - 2)$,  there exists an $ \widetilde{\e}_0 > 0$ 
such that for any subinterval $I \subset [\Sigma_{\rm inf}, \Sigma_{\rm inf}  + \widetilde{\e}_0]$ 
and any compact subset $K \subset \R^2$,
we have
\be \label{eq:dynloc}
\bE\left\{ \sup_t \left\| | X|^p e^{ - i H(A) t} {\bf 1}_I(H(A)) \chi_K \right\| \right\} < \infty .
\ee
\end{theorem}
We will use the notation that ${\bf 1}_S$ as well as   $\chi_S$  denotes the 
characteristic function of a set $S$.

\medskip

\noindent
{\it Remark.}
 We note that if $K_2 =0$, then no large $b_0$ assumption is  necessary, that is, the assertion 
of the theorem holds for any  $b_0\ge 2$. 
 Now $K_2 =0$ 
holds provided $B_{\rm det}$ and  $V$ are constant and  $U =  1$.
The condition  $U =  1$ can be realized for 
example  as follows. We 
choose  $\varphi \in C^\infty_0(\R^2;[0,1])$ with
$\varphi(x) = 0$,  if $ |x| \geq 1$, $\int \varphi = 1 $,
and  set, for $s > 0$,
$ u =  {\bf 1}_{\{ | x |_\infty \leq 1/2 \}} \ast \varphi_s$ and  $\varphi_s(x) = s^{-2} \varphi(x/s)$.
Conditions \eqref{kcond} or
\eqref{kcond1} can be satisfied by taking $s$ sufficiently small or sufficiently large, respectively. 

\medskip

The next theorem provides estimates on the location of  the deterministic spectrum $\Sigma$ 
of $H(B_\omega)$, under the influence of
the random potential. It will be used in the proof of Theorem \ref{thm:loclization}.
To formulate it, we define  two specific configurations of the collection of random variables,
$\omega_+$ and $\omega_-$, by $(\omega_{\pm})_z^{(k)} := m_{\pm}^{(k)}$,  and we set
\be\label{Esup}
E_{\rm inf} := \inf_{x \in \R^2} \left[ B_{\omega_-}(x) + V(x) \right] , \quad
E_{\rm sup} := \inf_{x \in \R^2} \left[ B_{\omega_+} (x) + V(x) \right].  
\ee
Moreover, we will write $M_{\pm} = \sum_{k=0}^\infty m^{(k)}_\pm$.
Note that $|m^{(k)}_{\pm} | \leq \sigma^{(k)}$ and thus $|M_{\pm}| \leq (1 - e^{- \rho})^{-1}$.

\begin{theorem} \label{thm:estonspec} Suppose {\bf (A)}  holds, 
 and let  $\rho > \ln 2$. 
Then  the following statements hold:
\begin{itemize}
\item[(a)] We have 
\be \label{eq:largeb0reg}
   E_{\rm inf}      \leq      \Sigma_{\rm inf}  \leq E_{\rm inf} + 4 K_2^2b_0^{-2} + 
\min(K_2 b_0^{-1/2} ,  K_3 b_0^{-1} ) ,
\ee
where we defined 
\be \label{eq:defofK3}
K_3  := 
2 \max_{|\alpha|=2} \left\{  \| D^\alpha B_{\rm det} \|_\infty +
 ( 1 - 4 e^{-\rho})^{-1} \| D^\alpha U \|_\infty + \| D^\alpha V \|_\infty \right\} ,
\ee 
if  $4 e^{-\rho} < 1$, and  $K_3 := \infty$ otherwise.
\item[(b)]  We have $E_{\rm inf} + \mu c_u (M_+ - M_-) \leq E_{\rm sup}$. 
\item[(c)] If 
$\Sigma_{\rm inf}  <   E_{\rm sup}$, then 
\be  \label{eq:estonspec}
 \Sigma  \supset
 [ \Sigma_{\rm inf}    ,   E_{\rm sup} ] .
\ee
 \item[(d)] In the special case when  $U  = 1$, $B_{\det}$ is constant and $V=0$, 
then $\Sigma_{\rm inf} =  B_{\rm det} + \mu  M_- $ and
 $$
\Sigma  \supset  \bigcup_{n \in \N_0} \{ ( 1 + 2 n ) ( B_{\rm det} + \mu [ M_- , M_+]  ) \} .
$$
\end{itemize}
\end{theorem}

\medskip

\noindent
{\it Remark.}
The finiteness of $K_3$ improves the upper bound on $\Sigma_{\rm inf}$ in the 
large $b_0$ regime, see \eqref{eq:largeb0reg}, but it requires higher regularity on the data.
We also remark that in view of (a) and (b) the condition $\Sigma_{\rm inf} <  E_{\rm sup}$ 
in (c) can be guaranteed if $c_u > 0$ and $b_0$ is sufficiently large.

\medskip

The paper is organized as follows.
In Section \ref{sec:methods} some previous methods to obtain a Wegner estimate are presented.
Sections \ref{sec:proofwegner}--\ref{sec:reg} are devoted to the proof of the Wegner estimate
as stated in Theorem \ref{thm:wegner}.  Its proof is given in Section  \ref{sec:proofwegner}
modulo the key  Proposition \ref{prop:lower}, whose proof is
given in Section \ref{sec:pr}. Section \ref{sec:reg}
contains some elliptic regularity estimates needed in Section \ref{sec:pr}.
The ergodicity property needed to show Theorem \ref{thm:deterministic}
will be given in Appendix A.
In Sections  \ref{sec:inspecbound}-\ref{sec:multiscale} we explain how the Wegner estimate leads to Anderson
localization.
In Section \ref{sec:inspecbound} an inner bound
on the deterministic spectrum is shown, i.e., a proof of Theorem \ref{thm:estonspec}  will be given.
In Section \ref{sec:initiallength}, an initial length scale estimate will be proven. This
estimate will then be used in Section \ref{sec:multiscale}, where the localization
result, Theorem \ref{thm:loclization}, will be shown.
We will use the multiscale analysis following the approach presented in Stollmann's book \cite{S}.
We remark that we could alternatively have followed the
setup presented by Combes and Hislop
in \cite{CH} to prove the initial length scale estimate by verifying 
their Hypothesis $[H1](\gamma_0, l_0)$.

We will use the convention that unspecified  positive constants only depending on $K_0$ and
$K_1$  are denoted by $C, C_0, C_1,...$
or $c,c_0,c_1,...$ whose precise values are irrelevant and may change from line to line.

\section{Main ideas of the proof of the Wegner estimate}\label{sec:methods}

The standard approach to prove  Wegner estimate for
random external potential is to use monotonicity of the
eigenvalues as a function of the random coupling parameters
(see, e.g. \cite{S} for an exposition).
Consider the simplest Anderson model of the form
$H_L=-\Delta + V_\om(x)$  with Dirichlet boundary
conditions on $\Lambda_L$. The random potential is given by
\be
V_\om(x)=\sum_{z\in \bbZ^d} \om_z u(x-z)
\label{vom}
\ee
  with  i.i.d. random variables
  $\om=\{ \om_z, z\in \bbZ^d\}$  and with  a local potential profile
function $u(x):\bR^d\to \bR$. By the first order perturbation formula for
any
eigenvalue $\lambda$ with normalized eigenfunction $\psi$ we
have
\be
     \frac{\partial \lambda}{\partial \om_z} = \langle\psi, u(\cdot -
z)\psi
   \rangle  =\int |\psi(x)|^2u(x-z)\rd x.
\label{pl}
\ee
We define the vector field $Y= \sum_{z\in \Lambda}  (\partial/\partial
\om_z)$
  on the space of the random couplings $\om$, where the summation is
over all $z\in \Lambda : = \Lambda_L\cap \bbZ^d$.
If, additionally, $\sum_z u(x-z)\ge c$ with some positive constant $c$,
then
$Y\lambda\ge c$. 
This estimate guarantees that each eigenvalue moves
with a positive speed as the random couplings vary in the direction of
$Y$.
In particular if $\om_z$ are continuous random variables
with some  mild regularity condition on  their
density function $v_z(\om_z)$
  then no eigenvalue can stick to any fixed energy $E$
when taking the expectation.

More precisely, if $\chi=\chi_{E,\eta}$ is
the characteristic function of the spectral interval $I=[E-\eta/2,
E+\eta/2]$
and $F$ is its antiderivative, $F'=\chi$, with $F(-\infty) = 0$, then
the expected number of eigenvalues in $I$ is estimated by
\be
   \bE \, \Tr\, \chi(H_L) \leq c^{-1} \bE \,  \Tr Y F(H_L)
  = c^{-1} \int_{\bR^{\Lambda}}
   \Big(\prod_{z\in \Lambda} v_z(\om_z) \rd \om_z\Big) \sum_{z\in \Lambda}
   \frac{\partial}{\partial \om_z} \Tr F(H_L).
\label{we}
\ee
If $v_z$ is sufficiently regular, then,
after performing an integration by parts
and using that $0\leq F\leq \eta$ together with
some robust Weyl-type bound for the number of eigenvalues,
one obtains
the Wegner estimate.
Note that the
proof essentially used that $\sum_z u(x-z)\ge c>0$,
in particular it does not apply to sign indefinite potential
profile $u$. We remark that for a certain class of random displacement 
models a different mechanism of monotonicity has been established in \cite{KLNS}
to prove the a Wegner estimate and Anderson localization. 

\medskip

For random vector potential of the form \eqref{Akl}, the
first order perturbation formula gives
\be
    \frac{\partial\lambda}{\partial \om_z} =
  \langle u(\cdot -z), j_\psi\rangle,
\label{fir}
\ee
where $j_\psi = 2 \mbox{Re} \; \bar\psi (p-A)\psi$
is the current of the eigenfunction. Unlike the non-negative
density $|\psi(x)|^2$, the
current is a vector and no apparent condition
on $u(\cdot - z)$ can guarantee that $Y\lambda \ge c>0$
for some $\psi$-independent vectorfield
of the form $Y= \sum_z c_z(\om) (\partial/\partial \om_z)$.

The method of \cite{HK} addresses the issue of
the lack of positivity of $Y\lambda$ for
both the sign non-definite
random potential \eqref{vom} case and the random vector potential
\eqref{Akl} case but it does not seem to apply for 
 random magnetic fields \eqref{rmf} due to the
long-range dependence of $A_\om$ generating $B_\om$.
 Moreover, it
uses the Birman-Schwinger kernel, i.e. it is restricted
for energies below  the spectrum of the deterministic part
$H_{\rm det}$ of the total Hamiltonian.


\bigskip

To outline our approach, we go back to \eqref{fir}, and will exploit
that
\be
\sum_z  \Big(\frac{\partial\lambda}{\partial \om_z}\Big)^2 =
   \sum_z|\langle u(\cdot -z), j_\psi\rangle|^2
\label{new}
\ee
is non-negative, and, in fact, it has an effective positive lower bound
(Proposition \ref{prop:lower}). The proof relies on three observations.
First,  $\| j_\psi\|_2$ has an effective lower bound because we
assume that there is a strictly positive background magnetic field
(Lemma \ref{lm:jnorm}).
Second,  $\|\nabla j_\psi\|_2$ has an upper
bound following from elliptic regularity (Lemma \ref{lm:jder}).
This will ensure that most of the $L^2$-norm of  $j_\psi$ comes
from low momentum modes.
  Finally, assuming that the random magnetic
field has modes on arbitrarily short scales, i.e. the summation
over $z$ in \eqref{new} is performed on a fine lattice, we see
that a substantial part of the low modes of $j_\psi$ is captured
by the right hand side of \eqref{new}, giving a positive
lower bound $c$ on \eqref{new}.

Using this lower bound we can estimate,
similarly to \eqref{we},
$$
   \bE \, \Tr\, \chi(H_L)  =\sum_\ell \chi(\lambda_\ell)
\leq c^{-1} \bE \,  \sum_z \sum_\ell (Y_z\lambda_\ell)^2
  \chi(\lambda_\ell),
$$
where $Y_z=(\partial/\partial \om_z)$ and $\lambda_\ell$ are
the eigenvalues of $H_L$. The square of the derivative,
$(Y_z\lambda_\ell)^2$,
  can be estimated in terms
of the second derivatives of the eigenvalues (see \eqref{main2}). 
By usual perturbation theory, to compute second derivatives of the
eigenvalues
requires first derivatives of eigenfunctions which seems
to be a hopeless task in case of possible multiple or near-multiple
eigenvalues. However, a key inequality in Lemma \ref{lm:trick}
ensures that the sum of second derivatives can be estimated
by the trace of the second derivative of the Hamiltonian itself.
  Since the Hamiltonian is quadratic
in the random parameters, this latter quantity can be computed.

\section{Proof of the Wegner estimate}  \label{sec:proofwegner}

\bigskip

In the following proof we
consider $L$ fixed. Set
  $k=k(L)$  such that
\be
\frac{1}{2}L^{-K}\leq 2^{-k}\leq L^{-K}
\label{ep}
\ee
with some fixed exponent $K$ to be determined later.
For brevity, we denote $\e=\e(L):= 2^{-k(L)}$.
Set $\Lambda_\e=(\e \bbZ)^2\cap \Lambda_{L+1} = \Lambda^{(k)}\cap \Lambda_{L+1}$.
Note that $|\Lambda_\e|\le CL^2\e^{-2}$.
For this given $L$, we decompose the magnetic field \eqref{B} as follows
$$
    B = \wt B + \mu B^{(k)}, \qquad  \wt B: = B_{\rm det} +
\mu \sum_{m=0\atop m\neq k}^\infty B^{(m)} .
$$
We will use only the random variables
in $B^{(k)}$ and we fix all random variables in
$B^{(m)}$, $m\neq k$, i.e. we consider $\wt B$ deterministic.
We will choose a divergence free gauge for $\wt B$, i.e.
$\nabla\times \wt A = \wt B$, $\nabla\cdot  \wt A =0$.
Since $k$ is fixed, we can drop the $k$ superscript
in the definitions of $B^{(k)}_z$, $\om^{(k)}_z$
$\beta^{(k)}_z$, $\sigma^{(k)}_z$ and $\sigma^{(k)}$, i.e.
$$
    B^{(k)}(x) =\sum_{z\in \Lambda_\e} B_z(x), \qquad
  B_z(x) = \om_z\beta_z(x), \qquad
     \beta_z(x) = u((x-z)/\e).
$$

We define two different vector potentials for $B_z$ by setting
$$
    a_z^{(1)} := \om_z \al_z^{(1)}, \qquad a_z^{(2)} := \om_z\al_z^{(2)}
$$
with
\be
    \al_z^{(1)}(x_1, x_2):= 
  \Big(\int_{-\infty}^{x_2} \beta_z(x_1, s)\rd s\Big){\bf e}_1, \qquad
\label{def:alpha}
     \al_z^{(2)}(x_1, x_2): =  - \Big(\int_{-\infty}^{x_1}
\beta_z(s, x_2)\rd s\Big){\bf e}_2,
\ee
where ${\bf e}_1=(1,0)$, ${\bf e}_2=(0,1)$ are the standard
unit vectors.
  Then $\nabla\times \al_{z,1}=\nabla\times \al_{z,2}=
  \beta_z$ and
  $\nabla\times  a_z^{(1)}=\nabla\times  a_z^{(2)}=
  B_z$ and notice that
\be
    \|\al_z^{(\tau)}\|_\infty\leq \e, \qquad \tau=1,2,
\label{alphabound}
\ee
and, actually, under condition \eqref{kcond1} we even have
$\|\al_z^{(\tau)}\|_\infty\leq \delta\e$.
  Let
\be
    A^{(1)}: =\wt A + \mu\sum_{z \in \Lambda_{\e}} \om_z \al_z^{(1)} , \qquad
    A^{(2)}: =\wt A + \mu\sum_{z \in \Lambda_\e}  \om_z \al_z^{(2)}
\label{vp}
\ee
then $\nabla\times A^{(\tau)}= B$, $\tau=1,2$.
We consider the two  unitarily equivalent random Hamiltonians
$$
    H_L( A^{(\tau)}): = (p- A^{(\tau)})^2 + V , \qquad \tau=1,2
$$
with Dirichlet boundary conditions on $\Lambda_L$.
For a while we will neglect the $\tau=1, 2$  indices;
all arguments below
hold for both cases.

\medskip

Let $\lambda$ be an eigenvalue of $H_L(A)$ with eigenfunction
$\psi$.
We consider $\lambda$ as a function of the collection
of random variables $\{ \om_z\}$.
For each fixed $z$,
\be
   Y_z\lambda:= \frac{\partial \lambda}{\partial \om_z}
  = 2\mu\mbox{Re} \int \bar \psi
   \al_z\cdot (p-A)\psi = \mu\int \al_z\cdot j_\psi ,
\label{ylambda}
\ee
where
$j_\psi=j = (j_1, j_2) = 2 \mbox{Re} \; \bar\psi (p-A)\psi$
is the current of the eigenfunction. Short calculation shows 
 that $j$ is gauge invariant and divergence free.

\bigskip

\noindent
{\bf  Proof of Theorem \ref{thm:wegner}.} 
  We first introduce a smooth high energy cutoff. Define 
the function 
$$
     t(u): = u \frac{ s^3 }{( s + u )^3}, \qquad \mbox{with}\quad s := 10 K_1 b_0,
$$
and the operator $T := t( H_L(A) )$. Clearly
\be
   \Tr \chi_{E,\eta}(H_L(A)) \leq  \Tr \chi_{t(E),\eta}(T)
\label{trtr}
\ee
since the derivative $t'$ is bounded by 1. In the sequel we set
$\chi = \chi_{t(E),\eta}$. Let $F(u)$ such that $F'=\chi$
with $F(u)=0$ for $u\leq t(E)-\eta/2$ and let $G(u)$ such
that $G'=F$ and $G(u)=0$ for $u\leq t(E)-\eta/2$.

Let $\lambda_1,\lambda_2,...$ denote the eigenvalues of $H_L(A)$ and let 
 $\tau_\ell = t(\lambda_\ell)$ be the   eigenvalues of $T$.
In Section~\ref{sec:pr} we will prove the following key technical estimate: 
\begin{proposition}\label{prop:lower}
  With the notations above, and assuming
   $\rho \ge\ln 2$ (i.e. $\sigma\leq \e$) there exist  positive constants $C_0$ and  $C_1$,
depending only on $K_0$ and $K_1$, and  a constant $L_0^*$,
depending on $K_0$, $K_1$, and $\delta$, such that 
for any $0 < \kappa \leq 1$ and  $a =  C_1 \kappa^{-1}$
\be
   \sum_z \Big(\frac{\partial \lambda_\ell}{\partial \om_z}\Big)^2
  = \sum_z (Y_z\lambda_\ell)^2 \ge C_0^{-1}L^{-a}\e^2\mu^2
\label{squarelow}
\ee
for any eigenvalue $\lambda_\ell$ of $H_{L}(A)$ and all  $L \geq L_0^* b_0^{\kappa}$.
\end{proposition}

{F}rom \eqref{squarelow} it easily follows that
\be
   \Tr \chi(T) = \sum_\ell \chi(\tau_\ell)
\leq CL^a\e^{-2} \mu^{-2}
  \sum_\ell \sum_z (Y_z\tau_\ell)^2 \chi(\tau_\ell) ,
\label{main1}
\ee
since  $Y_z \tau_\ell = g'(\lambda_{\ell}) Y_z \lambda_\ell$ 
and  for $\tau_\ell$ in the  support of $\chi$ the number 
$|g'(\lambda_{\ell})|$ is bounded from below by a universal constant.

Notice that for any $Y=Y_z$ and any $\ell$ we have
$$
Y^2 G(\tau_\ell) =  Y\big( (Y\tau_\ell) F(\tau_\ell)\big)
=  (Y^2 \tau_\ell)F(\tau_\ell) +  (Y\tau_\ell)^2
\chi(\tau_\ell) .
$$
Thus
\be
\begin{split}
    \sum_\ell \sum_z (Y_z\tau_\ell)^2\chi(\tau_\ell)
   = &\sum_\ell \sum_z  Y^2_z G(\tau_\ell) -
   \sum_\ell \sum_z   (Y_z^2 \tau_\ell)F(\tau_\ell)\cr
\label{main2}
  = &\sum_z  \Tr Y^2_z G(T) -
   \sum_\ell \sum_z   (Y_z^2 \tau_\ell)F(\tau_\ell) .
\end{split}
\ee

\begin{lemma}\label{lm:trick}
We have for any $Y=Y_z$
\be
    \Tr (Y^2 T) F(T)\le  \sum_\ell(Y^2\tau_\ell)F(\tau_\ell) .
\label{trick}
\ee
\end{lemma}

\noindent
{\bf Proof of Lemma \ref{lm:trick}.}
We use spectral decomposition, $T= \sum_\al \tau_\al
|u_\al\rangle \langle u_\al|$,
\be
\begin{split}
      \Tr (Y^2& T) F(T)
  = \sum_\al F(\tau_\al)\langle u_\al | \Big( Y^2\sum_\beta
   \tau_\beta |u_\beta\rangle\langle u_\beta|\Big) |u_\al\rangle
\cr
   = & \sum_\al F(\tau_\al)\langle u_\al | \Big( \sum_\beta
   (Y^2\tau_\beta) |u_\beta\rangle\langle u_\beta|
   + 2\sum_\beta (Y\tau_\beta) Y(|u_\beta\rangle\langle u_\beta|)
  + \sum_\beta \tau_\beta Y^2(|u_\beta\rangle\langle u_\beta|)\Big)
  |u_\al\rangle
\cr
    =& \sum_\al F(\tau_\al) (Y^2\tau_\al)
   + 2\sum_{\al,\beta} F(\tau_\al) (Y\tau_\beta)  \langle u_\al  |
  Y(|u_\beta\rangle\langle u_\beta|)
  |u_\al\rangle + \sum_{\al,\beta} F(\tau_\al)\tau_\beta
   \langle u_\al  |
  Y^2(|u_\beta\rangle\langle u_\beta|)
  |u_\al\rangle .
\label{tri}
\end{split}
\ee
The second term is zero, since
\be
\begin{split}
   \langle u_\al  |
  Y(|u_\beta\rangle\langle u_\beta|)
  |u_\al\rangle  = & \langle u_\al |Yu_\beta\rangle\langle u_\beta
  |u_\al\rangle + \langle u_\al |u_\beta\rangle\langle Yu_\beta
  |u_\al\rangle\cr
   =& \delta_{\al\beta}\Big( \langle u_\al |Yu_\al\rangle
+ \langle Y u_\al |u_\al\rangle\Big)=\delta_{\al\beta}Y \langle u_\al
  |u_\al\rangle =0
\end{split}
\ee
since $ \langle u_\al|u_\al\rangle=1$.
In the last term in  \eqref{tri}, we use that
$$
     0= Y\langle u_\al|u_\beta\rangle = \langle Yu_\al|u_\beta\rangle
   +  \langle u_\al|Yu_\beta\rangle
$$
and differentiating it once more:
$$
    0 = Y\Big( \langle Yu_\al|u_\beta\rangle
   +  \langle u_\al|Yu_\beta\rangle \Big)
   = \langle Y^2u_\al|u_\beta\rangle
   +  2\langle Yu_\al|Yu_\beta\rangle +
     \langle u_\al|Y^2u_\beta\rangle.
$$
Thus
\be
\begin{split}
    \langle u_\al  |
  Y^2(|u_\beta\rangle\langle u_\beta|)
  |u_\al\rangle
= & \langle u_\al  | \Big( |Y^2u_\beta\rangle\langle u_\beta|
   + 2|Yu_\beta\rangle\langle Yu_\beta| + |u_\beta\rangle\langle
Y^2u_\beta|
\Big) |u_\al\rangle \cr
   = & \delta_{\al\beta} \Big(  \langle u_\al  | Y^2u_\al\rangle
    +  \langle Y^2 u_\al  | u_\al\rangle\Big)
  + 2 |\langle Yu_\beta| u_\al\rangle|^2 \cr
    = & 2 |\langle u_\beta| Y u_\al\rangle|^2 - 2 \delta_{\al\beta}
    \langle Y u_\al| Yu_\al\rangle.
\end{split}
\ee
So for the last term in  \eqref{tri},
\be
\begin{split}\nonumber
  \sum_{\al,\beta}  F(\tau_\al)\tau_\beta
   \langle u_\al  | &
  Y^2(|u_\beta\rangle\langle u_\beta|)
  |u_\al\rangle  \cr
= &   2\sum_{\al,\beta} F(\tau_\al)\tau_\beta
   \Big( |\langle u_\beta| Yu_\al\rangle|^2 -  \delta_{\al\beta}
    \langle Y u_\al| Yu_\al\rangle\Big)\cr
   = &  2\sum_{\al} F(\tau_\al)\tau_\al
   \Big( \sum_\beta |\langle u_\beta| Yu_\al\rangle|^2 -
    \langle Y u_\al| Yu_\al\rangle\Big)
  +  2\sum_{\al,\beta} F(\tau_\al)(\tau_\beta-\tau_\al)
    |\langle u_\beta| Yu_\al\rangle|^2.
\end{split}
\ee
The first term is zero since $u_\beta$ is an orthonormal basis.
In the second term we use that $ |\langle u_\beta| Yu_\al\rangle|^2$
is symmetric in the $\al,\beta$ indices and write
\be
\begin{split}
\sum_{\al,\beta} F(\tau_\al)(\tau_\beta-\tau_\al)
    |\langle u_\beta| Yu_\al\rangle|^2
  = & \sum_{\al<\beta} \Big[ F(\tau_\al)(\tau_\beta-\tau_\al)
   +  F(\tau_\beta)(\tau_\al-\tau_\beta)\Big]
  |\langle u_\beta| Yu_\al\rangle|^2 \cr
   = & - \sum_{\al<\beta}  \big[ F(\tau_\al) - F(\tau_\beta)\big]
    (\tau_\al-\tau_\beta)
  |\langle u_\beta| Yu_\al\rangle|^2 \leq 0
\end{split}
\ee
since $F$ is monotone increasing.
This proves Lemma \ref{lm:trick}. $\Box$

\bigskip

Thus combining \eqref{main1}, \eqref{main2} and \eqref{trick}, we have
\be
   \Tr \chi(T)\leq CL^a\e^{-2}\mu^{-2}
\sum_z\Big( \Tr Y_z^2 G(T) - \Tr (Y_z^2 T)F(T)\Big).
\label{com}
\ee

We compute $Y_z^2T$. First, to present the  idea, imagine that
we did not have the high energy cutoff, i.e. $T$ were simply
  $(p-A)^2 +V$. Then
$$
Y_z [(p- \wt A - \mu\sum_\zeta \om_\zeta \al_\zeta)^2 + V ]
= -\mu\al_z\cdot (p- \wt A- \mu\sum_\zeta \om_\zeta \al_\zeta)
- \mu(p- \wt A- \mu\sum_\zeta \om_\zeta \al_\zeta)\cdot \al_z
$$
and
$$
  Y_z^2 (p- \wt A- \mu\sum_\zeta \om_\zeta \al_\zeta)^2 =2\mu^2\al_z^2 ,
$$
thus, using $|F|\leq \eta$ and \eqref{alphabound}, we would have
to compute $ \big| \Tr (Y^2_z T)F(T)\big|
  \leq C\e^2 \eta \Tr {\bf 1}(T\ge E-\eta/2)$.
Unfortunately, this trace is unbounded.
The smooth high energy cutoff ensures the finiteness of the
trace and gives a bound $ C\e^2 \eta L^2$,
  but it makes the second derivative calculation more
complicated.

\smallskip

To make the argument more precise, we go back to $T=t(H_L(A))$  with 
eigenvalues $\tau_\ell = t(\lambda_\ell)$ where $\lambda_\ell$ are the eigenvalues of $H_L(A)$.
Then
$$
     \Tr F(T) = \sum_\ell F(\tau_\ell) \leq \eta\sum_\ell
    {\bf 1}\Big( t(\lambda_\ell) \ge t(E)-\eta/2\Big)
   \leq  \eta\sum_\ell
    {\bf 1}\Big( E-C\eta\leq \lambda_\ell\leq  E^* + C\eta\Big)
$$
where $E^*>E$ is the other root of the equation $t(E^*)=t(E)$.
Using $E \leq K_1 b_0 =\frac{s}{10}$  it is easy to see that 
$E^*\leq C(E+s) \leq C K_1 b_0$ and $|t'(E)|, |t'(E^*)|$
are bounded from below by a universal constant which
was used in the last inequality. Thus, by applying Weyl's bound
\eqref{LTineq}
on the number of eigenvalues below a fixed threshold, we obtain
\be
     \Tr F(T) \leq C b_0 \eta  L^2,
\label{trf}
\ee
where $C$ depends $K_1$.
We note that the Weyl bound holds for magnetic Schr\"odinger operators
as well, namely, for any $K>0$ we have
\be
\begin{split} \label{LTineq}
   \#\{ \lambda_\ell \le K\} \leq & \#\big\{ \mbox{eigenvalues of $(p-A)^2 + V -
2K{\bf 1}_{\Lambda_L}$ between $[-K, -2K]$} \big\}  \cr
   \leq & CK^{-1}   \Tr \big[(p-A)^2 + V - 2K{\bf 1}_{\Lambda_L} \big]_-\cr
\leq &CK^{-1}\int_{\Lambda_L} \big[2K{\bf 1}_{\Lambda_L} \big]^2
   = CKL^2 ,
\end{split}
\ee
with some universal constant $C$.
  Here $\Tr [h]_-$ denotes the sum of absolute
values of the negative eigenvalues
of the operator $h$ and we applied it
to $h=(p-A)^2 + V - 2K{\bf 1}_{\Lambda_L}$ with Dirichlet boundary conditions
on $\Lambda_L$. The last inequality is the Lieb-Thirring inequality
that holds for magnetic Schr\"odinger operators as well.

To compute $Y_z^2T$, we define the resolvent
$$
    R = \frac{1}{s + (p-A)^2 + V } ,
$$
where $(p-A)^2 + V$ is understood with Dirichlet boundary conditions.
We have
\be
\begin{split}
    Y \frac{(p-A)^2 + V }{[s + (p-A)^2 + V ]^3}
    = & [Y(p-A)^2] R^3 -
\sum_{k=1}^3  [(p-A)^2 + V ] R^k [Y(p-A)^2] R^{4-k}
\end{split}
\ee
and thus
\be
\begin{split} \label{eq520}
   Y^2 \frac{(p-A)^2 + V }{[s + (p-A)^2 + V]^3}  = &
[Y^2(p-A)^2]\frac{1}{[s + (p-A)^2 + V]^3} 
\cr &
 - 2\sum_{k=1}^3  [Y(p-A)^2]
   R^k [Y(p-A)^2]R^{4-k} \cr
& - \sum_{k=1}^3  [(p-A)^2 + V ]
   R^k [Y^2(p-A)^2] R^{4-k}\cr
& + 2 \sum_{k,\ell=1\atop k+\ell\leq 4}^4  [(p-A)^2 + V] 
   R^k [Y(p-A)^2] R^\ell
   [Y(p-A)^2] R^{5-k-\ell} . \cr
\end{split}
\ee
Let $P={\bf 1}( H_L(A)\leq E^*+C\eta)$ be the spectral projection.
Since $F(T)=0$ on the complement of $P$, we can insert $P$ as
$$
\Tr (Y^2_z T)F(T) =\Tr P(Y^2_z T)PF(T) .
$$
Using that
$$
   Y_z(p-A)^2 = -\mu\al_z\cdot(p-A)-\mu(p-A)\cdot \al_z,
\qquad Y_z^2(p-A)^2 = 2\mu^2\al_z^2,
$$
the estimates
$$
   \Big\| (p-A) \frac{1}{(s+ (p-A)^2 + V)^{1/2}}\Big\| \le C,
\qquad \big\| P [(p-A)^2 + V]\big\| \le E^*+C\eta \le CK_1b_0 \le Cs,
$$
the fact that $\| V\|_\infty\le s$
and the bound $|\al_z|\leq\e$ we can estimate the right hand side  
of   \eqref{eq520} and we obtain
$$
    \|P(Y^2_z T) P \|\leq C\e^2
$$
with $C$ depending on $K_1$.
Thus
\be
    \big| \Tr (Y^2_z T)F(T)\big| \leq C\e^2 \Tr F(T)
\leq C b_0 L^2 \e^2\eta
\label{y2}
\ee
using the positivity of $F(T)$ and the bound \eqref{trf}.

\bigskip

Recalling \eqref{com} and $|\Lambda_\e|\le CL^2 \e^{-2}$, we proved that
\be
   \Tr \chi(T)\leq  C \mu^{-2}   L^a \e^{-2} 
\sum_z\Tr ( Y_z^2 G(T) ) +  C b_0 \mu^{-2}  L^{ a + 4} \e^{-2}  \eta  
\label{com1}
\ee
under \eqref{squarelow}. After taking expectation
with respect to the collection
$$
\{ \om_z\; : \; z\in\Lambda_\e\}
= \{ \om_z^{(k)}\; : \;  z\in\Lambda^{(k)}\cap \Lambda_L\},
$$
  we integrate by parts
$$
    \bE \Tr Y_z^2 G(T) = \int \prod_{\zeta\in \Lambda_\e}
  v_\zeta(\om_\zeta)\rd \om_\zeta
\frac{\partial^2}{\partial \om_z^2}  \Tr G(T)
  = \int  \prod_{\zeta\neq z} v_\zeta(\om_\zeta)\rd \om_\zeta
  \int v_z''(\om_z) \rd \om_z \Tr G(T).
$$

To compute $\Tr G(T)$ we use that $G(u)\leq C\eta u$, and
that $t(u) \leq s^3  (s +u)^{-2}$, we have
\be
\begin{split}
    \Tr G(T)
   \leq  C\eta\, s^3  \Tr \frac{1}{[s  + (p-A)^2 + V]^2}
  \leq  C\eta L^2  s^3 \int_{\bR^2}\frac{ \rd p}{[\frac{9}{10}s + p^2 ]^2} \leq C\eta s^2 L^2
\end{split}
\ee
using the integral representation $\alpha^{-2} = \int_0^\infty t e^{-\alpha t} dt$, with $\alpha >0$, 
Feynman-Kac-It\^{o} formula, and the diamagnetic inequality.
Using \eqref{2der} and  \eqref{squarelow}, we get from \eqref{com1} that
$$
  \bE \Tr \chi(T)\leq  C b_0^2  L^{a+4}\eta \mu^{-2}\e^{-4}\sigma^{-2}
+ C b_0 \mu^{-2}  L^{ a + 4} \e^{-2}  \eta  .
$$
Considering the choice of $\e=2^{-k}\sim L^{-K}$,
\eqref{ep} and \eqref{lowbound1}, we get
$$
  \bE \Tr \chi(T)\leq C b_0^2  \eta\mu^{-2} L^{a+4+4K+2(\ln 2)^{-1}\rho}
$$
and together with \eqref{trtr}
this completes the proof of Theorem \ref{thm:wegner}. $\Box$

\section{Proof of Proposition \ref{prop:lower}}\label{sec:pr}

In this section we prove that the lower bound  \eqref{squarelow}
  on $(Y_z\lambda)^2$ holds for
any eigenvalue.  Fix $\ell$ and denote $\lambda=\lambda_\ell$.
Using \eqref{ylambda}, we have
$$
    \sum_z (Y_z\lambda)^2 = 4\mu^2 \sum_z |(\al_z^{(1)}, j_\psi)|^2=
   4\mu^2 \sum_z |(\al_z^{(2)}, j_\psi)|^2
$$
which we write as
\be
  \mu^{-2}\sum_z (Y_z\lambda)^2 = 2 \sum_z |(\al_z^{(1)}, j_\psi)|^2+
   2 \sum_z |(\al_z^{(2)}, j_\psi)|^2 .
\label{Yal}
\ee
In the sequel $(\cdot, \cdot)$ denotes the scalar product on
$L^2(\Lambda_L)$.
We will prove the following two lemmas:

\begin{lemma}\label{lm:jder}
   There are   positive $d'$ and  $g'$
and a constant $C=C(K_0, K_1)$ such that
\be
  \int_{\Lambda_L} |\nabla j_\psi|^2 \leq CL^{d'} b_0^{g'}
\label{derupper}
\ee
for all normalized eigenfunction $\psi$ with energy $E\leq K_1 b_0$.
{F}rom the proof, $d'=126$ and $g'=60$. Here we adopted the notation 
$| \nabla j_\psi|^2
= \sum_{l,k=1}^2 | \partial_k (j_\psi)_l |^2$.
\end{lemma}

\begin{lemma}\label{lm:jnorm}  There are positive $d''$ and $g''$
and a positive constant $c=c(K_0, K_1)$ such that
\be
   \int_{\Lambda_L} |j_\psi|^2 \ge c b_0^{-g''}L^{-d''}
\label{L2lower}
\ee
uniformly for all normalized eigenfunction $\psi$, with energy $E\leq
K_1 b_0$.
{F}rom the proof, $d''=100$ and $g''=46$.
\end{lemma}

First we show that from these two Lemmas, \eqref{squarelow} follows
if $\delta_0$ in \eqref{kcond} or \eqref{kcond1} is sufficiently small.
Let $N:= [\delta^{-2}]+1$, where $[\; \cdot \; ]$ denotes the integer
part, and
define the square
$$
     Q_z :=\big\{ x\in \bR^2 \; : \; |x-z|_\infty \leq N\e\big\}
$$
and the $L^2$-normalized vector-valued functions
$$
        M_z^{(1)}(x) = \frac{1}{2N\e}{\bf 1}_{Q_z}(x) {\bf e}_1
    \qquad
        M_z^{(2)}(x) = \frac{1}{2N\e}{\bf 1}_{Q_z}(x) {\bf e}_2.
$$
We set $\Lambda_\e' = (2N\e\bbZ)^2\cap \Lambda_{L+1}\subset \Lambda_\e$
to be a sublattice of $\Lambda_\e$. The $z$-indices of $M_z^{(j)}$,
$j=1,2$,
will always run over this sublattice $z\in \Lambda_\e'$.
We write
\be
\begin{split}
   M_z^{(1)} = &  \frac{1}{2N\e^2}\sum_{k=-N}^N
  (\al_{z+k\e {\bf e}_1-N\e{\bf e}_2}^{(1)}- \al_{z+k\e {\bf e}_1+N\e{\bf
e}_2}^{(1)})
    + \cE_z^{(1)} \\
\label{M1} 
   M_z^{(2)} = & \frac{1}{2N\e^2}\sum_{k=-N}^N
  (\al_{z+k\e {\bf e}_2-N\e{\bf e}_1}^{(2)}- \al_{z+k\e {\bf e}_2+N\e{\bf
e}_1}^{(2)})
    + \cE_z^{(2)}
\end{split}
\ee
with  errors $\cE_z^{(\tau)}$ that are defined by these equations.
Let $Q_z' = \big\{ x\in \bR^2 \; : \; |x-z|_\infty \leq
(N+\delta^{-1})\e\big\}$.
We will prove the following estimates at the end of this section.
\begin{proposition}\label{prop:ce} With the notations above, we have
\be
      \mbox{supp} \, \cE_z^{(\tau)}\subset Q_z', \quad \tau=1,2,
\label{cesupp}
\ee
\be
     \int |\cE_z^{(\tau )}|^2 \leq \Theta\delta , \quad \tau=1,2.
\label{ceest}
\ee
where $\Theta =100$ under the condition \eqref{kcond} and
$\Theta=20 + \|\nabla u_0\|^2$ under the condition \eqref{kcond1}.
\end{proposition}

Suppose that \eqref{squarelow} is wrong, then, by \eqref{Yal}, we have
$$
   \sum_{\tau=1,2} \sum_{z\in \Lambda_\e} |(\al_z^{(\tau)}, j)|^2
  \leq CL^{-a}\e^2
$$
after dropping the subscript $\psi$.
Then, in particular
\be
   \sum_{\tau=1,2}
\sum_{z\in \Lambda_\e'} |(M_z^{(\tau)}, j)|^2
  \leq CL^{-a}\e^{-2}
  +  2\sum_{\tau=1,2}\sum_{z\in \Lambda_\e'} |(\cE_z^{(\tau)}, j)|^2 ,
\label{cont}
\ee
and using \eqref{cesupp}, we get
\be
    2\sum_{\tau=1,2}\sum_{z\in \Lambda_\e'} |(\cE_z^{(\tau)}, j)|^2
   \leq  2\sum_{\tau=1,2}\sum_{z\in \Lambda_\e'} \| \cE_z^{(\tau)}\|^2
   \| j {\bf 1}_{Q_z'}\|^2 \leq 16\Theta\delta \| j\|^2.
\label{errr}
\ee

For $z\in \Lambda_\e'$ and $\tau=1,2$, we let
$$
    \langle j_\tau\rangle_z  := (2N\e)^{-2}
  \int_{Q_z} j_\tau
= (2N\e)^{-1}(M_z^{(\tau)}, j).
$$
We write
$$
    j_\tau= \sum_{z\in \Lambda_\e'}  \langle j_\tau
\rangle_z {\bf 1}_{Q_z} + J_\tau
   = \sum_{z\in \Lambda_\e'} (M_z^{(\tau)}, j) (2N\e)^{-1}{\bf 1}_{Q_z}
+ J_\tau,
$$
where $J_\tau$ is orthogonal  to all ${\bf 1}_{Q_z}$, $z\in \Lambda_\e'$.
Since the $(2N\e)^{-1}{\bf 1}_{Q_z}$ functions are orthonormal, we have
from \eqref{cont} and \eqref{errr} that
\be
\begin{split}
     \| j\|^2=\sum_{\tau=1,2}\int_{\Lambda_L} |j_\tau|^2
& =\sum_{\tau=1,2}\sum_{z\in \Lambda_\e'}
  |(M_z^{(\tau)}, j)|^2 + \sum_{\tau=1,2}\|J_\tau\|^2_2
\cr
   &\leq CL^{-a}\e^{-2} + 16 \Theta\delta  \| j\|^2+ \sum_{\tau=1,2}
  \|J_\tau\|^2_2.
\end{split}
\label{C1}
\ee
Choosing $\delta_0=(32 \Theta)^{-1}$, for any $\delta\leq \delta_0$
we get
$$
    \| j\|^2
   \leq CL^{-a}\e^{-2} + 2 \sum_{\tau=1,2}
  \|J_\tau\|^2_2.
$$

  However, by Poincar\'e inequality
$$
   \int_{\Lambda_L} |J_\tau|^2 = \sum_{z\in \Lambda_\e'}
  \int_{Q_z} |j_\tau- \langle j_\tau\rangle_z|^2
  \leq \sum_{z\in \Lambda_\e'} (CN\e)^2\int_{Q_z}
  |\nabla j_\tau|^2\leq C\delta^{-4}\e^2
\int_{\Lambda_L}|\nabla j|^2
   \leq C\delta^{-4}\e^2 L^{d'} b_0^{g'}
$$
by Lemma \ref{lm:jder}, where the last constant depends on $K_0, K_1$.
Thus, from \eqref{C1} and Lemma \ref{lm:jnorm},
we have
$$
     cL^{-d''} b_0^{-g''}
  \leq CL^{-a} \e^{-2} + C\delta^{-4}\e^2 L^{d'} b_o^{g'} 
   \leq CL^{-a+2K} + CL^{d'+4-2K} b_0^{g'} 
$$
where we used that $\e\sim L^{-K}$ from
\eqref{ep} and  we assumed $L\ge \delta^{-4}$.
Using $\sigma\leq \e$, we get
\be \label{eq:contr}
cL^{-d''} b_0^{-g''}  \leq CL^{-a+2K} + CL^{d'+4-2K} b_0^{g'} .
\ee
Choosing first $K$ such that $g' + g'' <   \kappa (2 K - d' - d'' - 4)$ and then $a$ such that 
$g'' <   \kappa ( a - 2 K - d'')$,   we see that \eqref{eq:contr} is a contradiction 
if $L \geq  ( 2 C c^{-1} b_0 )^{\kappa}$ and $L \geq \delta^{-4}$. $\Box$

\bigskip

{\bf Proof of Proposition \ref{prop:ce}.}
To see the support property \eqref{cesupp}, we note that
$\mbox{supp}\; u \subset [-\delta^{-1}, \delta^{-1}]^2$ under either
conditions \eqref{kcond}
or \eqref{kcond1} on the profile function $u$. Therefore
  $\mbox{supp}\; \beta_\zeta
\subset Q_z'$ for any $\zeta\in \Lambda_\e$ with $|\zeta- z|_\infty\leq
N\e$
and \eqref{cesupp} follows immediately from the definitions of
$\alpha_z^{(\tau)}$, $\tau=1,2$, see  \eqref{def:alpha}.

For the proof of \eqref{ceest} we distinguish between the
two alternative conditions \eqref{kcond} and \eqref{kcond1}.
If  $u(x)$ satisfies \eqref{kcond} then  $\cE_z^{(1)}(x) =0$
unless $x$ satisfies either  $(N-1)\e\leq |x-z|_\infty \leq (N+1)\e$ or
$|x_1-z_1- k\e -\frac{\e}{2}|\leq \delta\e$
for some $k\in \bbZ$ with $|k|\leq N+1$.
Therefore $\cE_z^{(1)}$ is nonzero on a set with measure at most
$20N\e^2+ 20(N\e)^2\delta$ and $\|\cE_z^{(1)}\|_\infty\leq 2(N\e)^{-1}$
by the support properties of $\alpha_z^{(1)}$, thus
  \eqref{ceest} follows from $N\ge\delta^{-1}$ for $\tau=1$.
The case $\tau =2$ is analogous.

If $u(x)$ satisfies \eqref{kcond1}, then we still
have $\|\cE_z^{(\tau)}\|_\infty\leq 2(N\e)^{-1}$ since in the
$k$-summation in
\eqref{M1} at most $2\delta^{-1}$ terms overlap
and $\|\alpha_z^{(\tau)}\|_\infty \leq \delta\e$.
We can use this $L^\infty$ bound
in the regime  $(N-\delta^{-1})\e\leq |x-z|_\infty \leq
(N+\delta^{-1})\e$, which
gives a contribution of at most $20\delta^{-1}N^{-1}\leq 20 \delta$ to
the integral in \eqref{ceest} as before. In the  complementary regime we
claim that
\be
    |\cE_z^{(\tau)}(x) |\leq \delta (N\e)^{-1} \|\nabla u_0\|_\infty\quad
\mbox{for}
  \quad |x-z|_\infty\leq (N-\delta^{-1})\e.
\label{smooth}
\ee
which would give a contribution of at most $\delta^2\|\nabla
u_0\|_\infty^2$
to the integral in \eqref{ceest}.

To see \eqref{smooth}, we introduce a new variable $v=\delta\e^{-1}(x-z)$,
with components $v=(v_1,v_2)$ and note that
$|x-z|_\infty\leq (N-\delta^{-1})\e$ implies $|v|_\infty\leq N\delta-1$.
Using \eqref{kcond1}, \eqref{betadef}, \eqref{def:alpha} and \eqref{M1}
and after changing variables  we have
\be
\begin{split}
    \cE_z^{(1)}(x) = &\frac{1}{2N\e}\Bigg[ {\bf 1}_{Q_z}(x)
  -\sum_{k=-N}^N \int_{-\infty}^{v_2}
  \delta \Big( u_0( v_1- k\delta, s+N\delta) - u_0(v_1-k\delta,
s-N\delta)\Big)\rd s
  \Bigg]{\bf e}_1
\cr
= & \frac{1}{2N\e}\Bigg[1-\delta\sum_{k=-N}^N
\int_{-1-N\delta}^{1-N\delta}
   u_0( v_1- k\delta, s+N\delta) \rd s\Bigg]{\bf e}_1,
\end{split}
\ee
where we used that $u_0$ is supported on $[-1,1]$ and that
$|v|_\infty\leq N\delta-1$ implies $1-N\delta\leq v_2\leq N\delta-1$
  to restrict the regime
of integration for the first term and to conclude that the second
integrand is zero.
Since $|v|_\infty\leq N\delta-1$, we see that $|v_1-k\delta|\ge 1$ if $|k|
>N$,
thus the summation can be extended to all $k\in \bbZ$  without changing the
value
of the right hand side, since $u_0$ is supported on $[-1,1]^2$.
We use the fact that for any $f\in C^1(\bR)$ function with compact support
$$
  \Big| \sum_{k\in \bbZ} \delta f(\delta k) -   \int_\bR f(t) \rd t \Big|
  \leq \delta \| f'\|_\infty |\mbox{supp} \, f|
$$
that can be easily obtained by Taylor expansion. Thus
$$
     |\cE_z^{(1)}(x)|
   \leq   \frac{\delta}{N\e} \| \partial_1 u_0\|_\infty.
$$
The proof for $\cE_z^{(2)}$ is analogous and
this completes the proof of Proposition \ref{prop:ce}. $\Box$.

\section{Proof of the regularity lemmas}\label{sec:reg}

Since $j_\psi$ is gauge invariant, to prove Lemmas
\ref{lm:jder} and \ref{lm:jnorm}, we can work in an appropriate
gauge $\wt A$ for the deterministic part $\wt B$ of
the magnetic field.
  Since $\psi$ is supported in $\Lambda_L$,
it is sufficient to construct $\wt A$ on $\Lambda_L$.

\begin{proposition} Given a bounded magnetic field
$\wt B$ on $\Lambda_L$,
there exists a vector potential $\wt A$, $\nabla\times \wt A=\wt B$,
  that is divergence free, $\nabla\cdot\wt A=0$, and for any $1<p<\infty$
\be
     \| \wt A\|_p\leq C_p L \| \wt B\|_\infty
\label{A*}
\ee
with some constant $C_p$ depending only on $p$. Here $\| \cdot \|_p$
denotes the $L^p(\Lambda_L)$-norm.
\end{proposition}

{\it Proof.}  Let $A^*$ be
the Poincar\'e gauge for $\wt B$, i.e.
$$
A^*_1(x)=-\int_0^1 t\wt B(tx)x_2\rd t, \qquad A^*_2(x)
=\int_0^1 t\wt B(tx)x_1\rd t,
$$
then clearly  $\| A^*\|_p\leq L\| \wt B\|_\infty$ and
$\nabla\times  A^* = \wt B$.
Define 
$$
w(x) := \frac{1}{2\pi} \int_{\Lambda_L} \log|x-y| A^*(y) d y ,
$$
i.e. $\Delta w = A^*$,
and $\phi := \nabla \cdot w$. Note that $\nabla \cdot \nabla \phi = 
\Delta \nabla \cdot w = \nabla \cdot \Delta w = \nabla A^*$. By the
Calderon-Zygmund inequality  (see  Theorem 9.9 of \cite{GT}) we have
$$
\|\nabla\phi\|_p=\| \nabla ( \nabla \cdot w ) \|_p \leq C_p \| \Delta w \|_{p}
 = C_p \| A^*\|_p  \leq C_p L \| \wt B \|_\infty
$$
 for any $1<p<\infty$. 
We define   $\wt A = A^* - \nabla\phi$,  then $\nabla\times\wt A=\wt B$,
  $\nabla\cdot\wt A=0$ and \eqref{A*} holds. $\Box$

\bigskip

\noindent
{\bf Proof of Lemma \ref{lm:jder}.}
Since $\psi$ is
a Dirichlet eigenfunction, we have
\be
  -\Delta \psi + 2iA\cdot \nabla \psi + i(\nabla\cdot A)\psi +
(A^2 + V -E )\psi=0
\label{efn}
\ee
and
$$
    \int |(\nabla-iA)\psi|^2 \le\widetilde{E}\|\psi\|^2 ,
$$
with $\widetilde{E} := E + \| V \|_\infty$.
Since $\psi$ is supported on $\Lambda_L$,
all integrals and norms in this proof will be  in $\Lambda_L$

By the Gagliardo-Nirenberg inequality
\be
\bs
    \|\psi\|_6 \leq   C\|\nabla\psi\|_{3/2}  \leq & C\|
(\nabla-iA)\psi\|_{3/2}
    + C\| A\psi\|_{3/2} \cr
   \leq  &  CL^{1/3} \|  (\nabla-iA)\psi\|_2 + \| A\|_6 \|\psi\|_2 \cr
    \leq & C(L^{1/3}\widetilde{E}^{1/2} + \|A\|_6) \|\psi\|_2
\end{split}\label{psi6}
\ee
and similarly
\be
    \| \psi\|_4 \leq  C\|\nabla\psi\|_{4/3} \leq  C(L^{1/2}\widetilde{E}^{1/2} + \|A\|_4)
\|\psi\|_2.
\label{psi4}
\ee
Then
\be
    \| \nabla\psi\|_2 \leq  \| (\nabla-iA)\psi\|_2 + \| A\psi\|_2 \leq
    \widetilde{E}^{1/2} \|\psi\|_2+\|A\|_4\|\psi\|_4
   \leq  C \big( \widetilde{E}^{1/2} + L^{1/2}\widetilde{E}^{1/2}\|A\|_4 +\|A\|_4^2 \big)\|\psi\|_2.
\label{Dpsi}
\ee
We  use \eqref{efn} and Calderon-Zygmund inequality in the form given in Corollary 9.10 of \cite{GT} for 
$\psi \in W_0^{2,p}(\Lambda_L)$
\be
\begin{split}
   \|D^2\psi\|_4 \leq & C \|\Delta\psi\|_4 \leq C \left( 2\|A\nabla\psi\|_4 +
   \|\nabla\cdot A\|_\infty \|\psi\|_4
   + \|  (A^2 + V -E)\psi\|_4 \right) \cr
  \leq & C \left( 2 \|A\|_{20} \|\nabla\psi\|_5 + \big[\|\nabla\cdot A\|_\infty
+\widetilde{E}\big]\|\psi\|_4
   + \|A\|_{24}^2\|\psi\|_6 \right) .
\label{CZ}
\end{split}
\ee
We can estimate
\be
\begin{split}
     \|\nabla\psi\|_5 & \leq C
\|\nabla\psi\|_6^{9/10}\|\nabla\psi\|_2^{1/10}
   \leq C\kappa \|\nabla\psi\|_6 + C\kappa^{-9}\|\nabla\psi\|_2\cr
   & \leq C\kappa \|D^2\psi\|_{3/2} + C\kappa^{-9}\|\nabla\psi\|_2\cr
& \leq C\kappa L^{5/6}\| D^2\psi\|_4 + C\kappa^{-9}\|\nabla\psi\|_2 
\label{nabla4}
\end{split}
\ee
for any $\kappa>0$,
where we used the Gagliardo-Nirenberg inequality once more.
Choosing $\kappa = (4CL^{5/6}\|A\|_{20})^{-1}$ we can
absorb the first term in the right hand side of \eqref{CZ} into
the left term and we obtain
$$
   \|D^2\psi\|_4 \leq C \left( L^{15/2}\| A\|_{20}^{9}\|\nabla\psi\|_2
+\big[\|\nabla\cdot A\|_\infty +\widetilde{E}\big]\|\psi\|_4
   + \|A\|_{24}^2\|\psi\|_6 \right) .
$$
The vector potential is given by \eqref{vp},
  $A =\wt A + \sum_{z\in \Lambda_\e} \om_z\al_z$,  with
$|\om_z|\leq \sigma$, $\| \al_z\|_\infty \leq \e$, so
$|A|\leq |\wt A| + C\sigma$ since among all $\al_z$ at most
$C\e^{-1}$ of them overlap. Thus, from \eqref{A*} we have
\be
    \| A \|_p \leq C_p L K_0 b_0 \qquad 1<p<\infty.
\label{Ab}
\ee
Moreover, $|\nabla\cdot A| \leq C\sigma/\e$
since $\nabla\cdot \wt A=0$ and $|\nabla\alpha_z|\leq C$.
Using these  estimates together with
  \eqref{psi6},  \eqref{psi4} and \eqref{Dpsi}, we have proved
\be
  \|D^2\psi\|_4
   \leq  CL^{39/2} b_0^{10} \|\psi\|_2
\label{ellreg1}
\ee
with $C= C(K_0, K_1)$. By H\"older inequality we also get
\be
    \| D^2\psi\|_2 \leq CL^{20} b_0^{10} \|\psi\|_2 .
\label{ellreg}
\ee
Going back to the estimate on $\|\nabla\psi\|_6$  used in \eqref{nabla4}, we
also have
\be
   \|\nabla\psi\|_6\leq  CL^{5/6}\|D^2\psi\|_{4}
    \leq  CL^{21} b_0^{10} \|\psi\|_2
\label{egyder}
\ee
using $L\ge 1$,  and we have 
\be \label{egyder2}
\|\nabla\psi\|_4\leq L^{1/6} \|\nabla\psi\|_6 .
\ee
Moreover, from \eqref{psi6} and  \eqref{psi4}
\be
   \|\psi\|_6, \|\psi\|_4 \leq CL\|\psi\|_2.
\label{noder}
\ee

 {F}rom \eqref{ellreg1}
 and Sobolev inequality applied to $\nabla \psi$, we have 
$$
\| \nabla \psi \|_\infty \le \| D^2\psi \|_4 + \| \nabla \psi \|_4 <\infty.
$$
Then $j = 2\mbox{Re} \big[ -i\bar{\psi} \nabla \psi - A|\psi|^2\big]$ vanishes at the boundary,
since  $\nabla \psi$ is bounded and  $\psi$ vanishes
at the boundary.
To prove \eqref{derupper},
we use that
$$
    \int |\nabla j|^2 = \int |\nabla\times j|^2
$$
since $\nabla\cdot j=0$ and $j$ vanishes on the boundary.
 Now we compute
$$
  |\nabla\times j|\leq 2|\nabla\times [ \bar \psi (p-A)\psi]|
   \leq C\Big(
|\nabla\psi|^2 + |D^2\psi||\psi|+ |A| |\psi||\nabla\psi| +
|B||\psi|^2\Big).
$$
Thus, using the estimates \eqref{Ab}, \eqref{egyder}, \eqref{egyder2} and \eqref{noder}
$$
    \int |\nabla j|^2\leq C\Big( \|\nabla\psi\|^4_4 + \|D^2\psi\|^4_4
   + \|\psi\|^4_4 + \|\nabla\psi\|^6_6 + \|\psi\|^6_6 +
  \|A\|^6_6 + \|B\|_\infty^2 \|\psi\|^4_4 \Big) \leq C L^{126}b_0^{60} 
\|\psi\|_2^4 .
$$
This bound proves  Lemma \ref{lm:jder}. $\Box$

\bigskip

\noindent
{\bf
Proof of Lemma \ref{lm:jnorm}.}
We first we need a lower bound on the eigenfunction.
Let $\psi$ be a normalized Dirichlet eigenfunction of $H_L(A)$, i.e.
$$
    \Big[(-i\nabla - A)^2 + V  - E\Big]\psi =0 .
$$
Let $x_0$ be the point where $|\psi(x)|$ reaches its maximum.
Since
$$
1 = \int_{\Lambda_L} |\psi|^2 \leq L^2 |\psi(x_0)|^2
$$
we have
\be
|\psi(x_0)|\ge \frac{1}{L}.
\label{psi0}
\ee
In particular, $x_0 \in {\rm int} ( \Lambda_L )$. 
Now we consider a disk $\wt D$ of radius $\ell$ about $x_0$, where $l > 0$ is sufficiently small so that
$\tilde{D} \subset \Lambda_L$. 
Let
$$
    \langle \psi\rangle : =\frac{1}{|\wt D|}\int_{\wt D} \psi, \qquad
    \langle \nabla\psi\rangle:=\frac{1}{|\wt D|}\int_{\wt D} \nabla\psi .
$$
Notice that from \eqref{Dpsi} and \eqref{Ab}
\be
   \langle \nabla\psi\rangle^2 \leq\frac{1}{|\wt D|}\int_{\wt D}
|\nabla\psi|^2
   \leq \ell^{-2}\|\nabla \psi\|_{L^2(\wt D)}^2\leq CL^4 b_0^4  \ell^{-2} .
\label{not}
\ee
For $x \in \wt D$ define  
$$
    f(x) :=\psi(x) -  \langle \psi\rangle - \langle \nabla\psi\rangle\cdot
(x-x_0) .
$$
Then
$$
   \int_{\wt D} |\nabla f|^2 =    \int_{\wt D}|\nabla \psi -
 \langle \nabla\psi\rangle|^2          \leq C\ell^2\int_{\wt D} |D^2\psi|^2 , 
$$
by Poincare inequality in $\wt D$.
Thus, applying Sobolev inequality  for $f$, we have
$$
     \| f \|_\infty 
\leq C\Big( \ell \| D^2f\|_{L^2(\wt D)}
   +\ell^{-1}\|f\|_{L^2(\wt D)}\Big) \leq C \ell \| D^2f\|_{L^2(\wt D)},
$$
where in the last step we used $\langle f \rangle=0$, $\langle \nabla f\rangle=0$, and we
used Poincar\'e inequality twice 
$$
    \| f \|_{L^2(\wt D)}\leq C\ell \| \nabla f\|_{L^2(\wt D)} 
\leq C\ell^2 \| D^2f\|_{L^2(\wt D)}.
$$
Thus
$$
    \| f \|_\infty \leq C\ell L^{20} b_0^{10} 
$$
by \eqref{ellreg} and so, by \eqref{not},
\be
\begin{split}
   |\psi(x)-\langle \psi\rangle|\leq & C\ell L^{20} b_0^{10}  + |x-x_0|
|\langle \nabla\psi\rangle|\cr \leq &
   C\ell L^{20} b_0^{10} + CL^2 b_0^2 |x-x_0|\ell^{-1}\cr
   \leq & CL^{11} b_0^6 |x-x_0|^{1/2}
\end{split}
\ee
after choosing $\ell =L^{-9}b_0^{-4} |x-x_0|^{1/2}$.
{F}rom \eqref{psi0}
this guarantees that there is a disk $D=D_R$ about $x_0$ of radius
$R= cL^{-24}b_0^{-12} $
so that
$|\psi(x)|\ge \frac{1}{2} |\psi(x_0)|$ for $x\in D$,
i.e.
\be
    |\psi(x)|\ge \frac{1}{2L}, \qquad x\in D.
\label{low}
\ee

Now we give a lower bound on the current.
On $D$ the wave function $\psi$ does not vanish, so
we can write it as $\psi= |\psi|e^{i\theta}$ with
some real phase function $\theta$. Then
$$
    \mbox{Re} \; \bar \psi (p-A)\psi =
    \mbox{Re} \; |\psi|e^{-i\theta} (-i\nabla -A)[e^{i\theta}|\psi|]
   = \mbox{Re} \; |\psi| (-i\nabla -(A-\nabla\theta))|\psi|
    =  -(A-\nabla\theta)|\psi|^2 .
$$
Thus
\be
    \int_{D} |j|^2 =  4\int_{D} (A-\nabla\theta)^2|\psi|^4 \ge cL^{-4}
    \int_D (A-\nabla\theta)^2,
\label{jj}
\ee
where we used \eqref{low}.
Finally, we will need the following elementary lemma:
\begin{lemma}\label{lm:gauge}
   Let $D=D_R$ be a disk of radius $R$ and let $A$ be a
vector potential generating $B$ with a lower bound $B(x) \geq b_0$.
Then
$$
    \int_{D_R} A^2\ge \frac{\pi}{8}b_0^2R^4.
$$
\end{lemma}

\noindent
{\bf Proof.}  Let $S_r$ be the circle of radius $r$ with the same center as $D$.
\be
\begin{split}
     \int_D A^2 & = \int_0^R \Big(\int_{S_r} A^2\Big)\rd r
   \ge  \int_0^R \frac{1}{2\pi r}\Big(\int_{S_r} A \rd s\Big)^2\rd r\nonumber\\
 & = \int_0^R \frac{1}{2\pi r}\Big(\int_{D_r} B\Big)^2\rd r
  \ge \int_0^R \frac{1}{2\pi r}\Big(\pi r^2 b_0 \Big)^2\rd r
  = \frac{\pi}{8}b_0^2R^4 .  \qquad \Box \nonumber
\end{split}
\ee

\smallskip\noindent
Lemma \ref{lm:jnorm} now follows from  \eqref{jj} and
 Lemma \ref{lm:gauge} 
$$
   \int_D |j|^2\ge cL^{-4}  \int_D (A-\nabla\theta)^2 \ge cL^{-4}b_0^2|D|^2 =
   cb_0^{-46} L^{-100}.  \qquad \qquad \Box
$$

\section{Deterministic Spectrum} \label{sec:inspecbound}

The goal of this section is to prove Theorem  \ref{thm:estonspec}. 
For $l > 0$ and $x \in \R^2$ we denote   by
$$
\Lambda_l(x) := \{ y \in \R^2 :   | y - x |_\infty < l/2 \}
$$
the open square of sidelength $l$ centered at $x$.
We introduce the constant
\begin{equation} \label{def:cdelta}
c_\delta = \left\{ \begin{array} {ll} \frac{3}{2}  , & {\rm in \ \ case } \ \  \eqref{kcond} \\
                  \delta^{-1}  ,  &    {\rm in \ \ case } \ \  \eqref{kcond1}
                   \end{array} \right.
\end{equation}
which gives the distance beyond which the random magnetic field is independent.
{F}rom Theorem \ref{thm:deterministic} recall that $\Sigma$ denotes the almost surely
deterministic spectrum.

\begin{theorem} \label{thm:ibspec}  Let {\bf (R)} and {\bf (i.i.d.)} hold
 with
$\rho > \ln 2$.
Assume that $B_{\rm det}$ and $V$ are $\Z^2$ periodic.
Then
$$
\Sigma \supset \bigcup_{L \in \N}
\bigcup_{ \omega  \, \in \, \mathcal{V}_L   }
\sigma( H( B_{\omega})  )  ,
$$
where $\mathcal{V}_L = \{ \,  \omega \in \Omega \, : \, \forall k \in \N  , \, 
 \omega_z^{(k)} \in [m_-^{(k)},m_+^{(k)}] , \,  \omega_{z + nL}^{(k)} = \omega_{z}^{(k)} , 
\, \forall n \in \Z^2  \, \}$ is the set of $L$-periodic configurations.
\end{theorem}

\begin{proof}  By unitary equivalence, we can fix a gauge.
 Given a magnetic $B$-field, for any $y \in \R^2$ we define the vector potential
$$
A_y[B](x) := \left( 0 , \int_{y_1}^{x_1} B(x_1',x_2) dx_1' \right) .
$$
Fix $\omega_0 \in \mathcal{V}_L$ for some $L \in \N$.
Let $E \in \sigma(H(B_{\omega_0}))$.
Since the magnetic Hamiltonian  is essentially self-adjoint on $C_0^\infty$ functions,
it follows that  there exists a normalized sequence $\varphi_n \in C_0^\infty$ such that
\begin{equation} \label{lem:trial:eq:1}
 \| ( H(A_0[ B_{\omega_0}]) - E) \varphi_n   \| \to 0  , \quad (n \to \infty) .
\end{equation}

Let $l_n \in 2 \N +1 $ be such that ${\rm supp}(\varphi_n ) \subset \Lambda_{l_n}(0) $ and $l_n \geq n$.
For $x \in \Z^2$, we introduce the following random  variables
\begin{align*}
B_{x,l,\omega_0}(\omega)  &:=    \| (   B_{\rm ran}^\omega -   B_{\rm ran}^{\omega_0} )
 \upharpoonright \Lambda_l(x)  \|_\infty   \\
B'_{x,l,\omega_0}(\omega)  &:=   \sum_{i=1,2} 
 \|  \partial_{x_i} ( B_{\rm ran}^{\omega} - B_{\rm ran}^{\omega_0} )  \upharpoonright \Lambda_l(x)  \|_\infty   ,
\end{align*}
and we  define the set
\begin{equation} \label{lem:trial:eq:2}
\Omega_n(x) = \Omega_{n,\omega_0}(x) := \{ \omega \in \Omega \, | \,
B_{x,l_n,\omega_0}(\omega) \leq  l_n^{-3} , B'_{x,l_n,\omega_0}(\omega) \leq  l_n^{-3} \} .
\end{equation}
Using the properties  of the random potential, it is straight forward to verify  that
  ${\bf P}(\Omega_n(x))$ is independent of $x$ and strictly positive. Moreover,
  if ${\rm dist}(\Lambda_{ l_n}(x),\Lambda_{ l_n}(y)) \geq 2 c_\delta$   then
$\Omega_n(x)$ and $\Omega_n(y)$ are independent. It now follows that 
 for a.e. $\omega \in \Omega$ there exists an
$x_n = x_{n,\omega,\omega_0} \in L \Z^2 $ such that $\omega \in \Omega_{n,\omega_0}(x_{n})$.
 We set
$$\widetilde{A} = \widetilde{A}_{n,\omega,\omega_0} := A_{x_n}[B_\omega] - A_{x_n}[B_{\omega_0}] =
A_{x_n}[\mu B_{\rm ran}^{\omega}  -  \mu B_{\rm ran}^{\omega_0}] . $$
Then setting $\varphi^{x_n}_n(\cdot ) = \varphi_n(\cdot - x_n  )$, we have
\begin{eqnarray} \label{eq:trialomega}
 \|  (H(A_{x_n }[ B_{\omega}]) - E )
 \varphi_n^{x_n} \|  \leq \| (H(A_{x_n }[ B_{\omega_0}]) - E ) \varphi_n^{x_n}   \|
 + 2 R_1 + R_2 + R_3 ,
\end{eqnarray}
with
\begin{align*}
R_1 &:=  \| \widetilde{A} \cdot ( p -
A_{x_n}[B_{\omega_0}] ) \varphi_n^{x_n}    \|   , \\
R_2 &:=   \|   \widetilde{A}^2  \varphi_n^{x_n}   \| , \\
R_3 &:=   \|  ( \nabla \cdot  \widetilde{A} ) \varphi_n^{x_n}  \|  .
\end{align*}
Let  $\chi_n$ denote the characteristic function of $\Lambda_{l_n}(x_n)$.
We estimate
$$
R_1 \leq
\| \widetilde{A}\chi_n \|_\infty       \|   ( p - A_{0}[B_{\omega_0}] ) \varphi_n      \|
\leq  C   l_n^{-2} l_n  (E+1)^{1/2} ,
$$
 where in the first inequality we used the $L$-periodicity  of $B_{\omega_0}$ and
in the second inequality we used   \eqref{lem:trial:eq:1} and the definition of $\widetilde{A}$.
Using again the definition of $\widetilde{A}$, we similarly find
$ R_2  \leq C l_n^{-1}$ and $R_3  \leq C l_n^{-2} .$
Using a gauge transformation  such that
$
H(A_{x_n }[ B_{\omega_0}] ) = e^{ i \lambda_{n,\omega,\omega_0}} H(A_{0}[ B_{\omega_0}] )
 e^{- i \lambda_{n,\omega,\omega_0} }
$
it now follows from \eqref{eq:trialomega} and  \eqref{lem:trial:eq:1} that
$$
 \lim_{n \to \infty} \|  (H(A_{0}[ B_{\omega}]) - E )  e^{- i \lambda_{n,\omega,\omega_0}} \varphi_n^{x_n} \|  = 0 .
$$
This yields the theorem.
\end{proof}

\begin{lemma} \label{lem:estoninfspec}
Let $\Lambda$ be a square or $\R^2$ and $B = \nabla \times  A$. As an inequality in the sense of forms
in  $L^2(\Lambda)$
$$
H_\Lambda(A) \geq \pm B  + V .
$$
\end{lemma}
\begin{proof} Let $\sigma_i$ denote the $i$-th Pauli matrix. Then for 
$\varphi \in C_0^\infty(\Lambda;\C^2)$ we have
\begin{eqnarray*}
\langle \varphi , (p - A)^2 \varphi \rangle  &  =
\left\langle \varphi , \left( \left[\sum_{i=1}^2 (p_i -  A_i ) \sigma_i \right]^2 
- \sigma_3 B \right) \varphi \right\rangle   \geq \langle \varphi ,
- \sigma_3 B \varphi \rangle .
\end{eqnarray*}
The lemma now follows by density argument.
\end{proof}

\medskip

\noindent
{\bf Proof of Theorem \ref{thm:estonspec}}.
(d).  First observe that $\Sigma_{\rm inf} \geq B_{\rm det} + \mu M_-$, by 
 Lemma \ref{lem:estoninfspec}. {F}rom Theorem  \ref{thm:ibspec} and the fact that a magnetic
Hamiltonian with a constant magnetic field is explicitly solvable, we find that
$$
\sigma( H(B_\omega)) \supset  \bigcup_{n \in \N_0} \left\{ ( 1 + 2 n ) ( B_{\rm det} + \mu I_v ) 
\right\}  ,
$$
where
$
I_v   = [M_-,M_+]$. (d) now follows.

(b). This follows directly from the definition of $E_{\rm inf}$ and $E_{\rm sup}$.

(a).
First observe that $\Sigma_{\rm inf} \geq E_{\rm inf}$ follows from
 Lemma  \ref{lem:estoninfspec}  and the definition of $E_{\rm inf}$.
Next we show that
\be
\label{eq:bottomofspec}
E_{\rm inf} + 4 K_2^2b_0^{-2} + 
\min(K_2 b_0^{-1/2} ,  K_3 b_0^{-1} )       \geq \Sigma_{\rm inf}
\ee
 using
a trial state.
By continuity and periodicity of $V +B_{\omega_{-}}$ we have
$$
E_{\rm inf}  = (V + B_{\omega_{-}})(\widetilde{x}) ,
$$
for some $\widetilde{x} \in \R^2$.
We choose  the gauge
$$
A_{{\rm inf}}(x)  = \frac{1}{2} \left(  - \int_{\widetilde{x}_2}^{x_2} B_{\omega_{-}}(x_1, y_2) \rd y_2 , 
   \int_{\widetilde{x}_1}^{x_1} B_{\omega_{-}}(y_1, x_2) \rd y_1        \right)
$$
and we set $A_{0}(x)  := \frac{1}{2}B_{\rm inf}^{(0)}( -(x_2 - \widetilde{x}_2), x_1- \widetilde{x}_1 )$ 
with $B_{\rm inf}^{(0)} :=  B_{\omega_{-}}(\widetilde{x})$.
Let us consider the trial state
\be
\varphi_0(x) = \exp(-\frac{1}{4}B_{\rm inf}^{(0)}|x - \widetilde{x}|^2  ) ,
\ee
which satisfies  $ \left[ ( p - A_{0} )_1 + i (p-A_{0})_2 \right]  \varphi_0 = 0 $.
Using a straight forward calculation we find
\begin{eqnarray}
 \langle  \varphi_0 , H(A_{\rm inf}) \varphi_0 \rangle &= &
\|  \left[ ( p - A_{\rm inf} )_1 + i (p - A_{\rm inf})_2 \right]  \varphi_0 \|^2 + 
 \langle  \varphi_0 , (B_{\omega_{-}} + V ) \varphi_0 \rangle   \label{eq:upbouninfspec21}  \\
&=& E_{\rm inf} \| \varphi_0 \|^2 + \| \left[ ( A_0 - A_{\rm inf})_1 
+ i ( A_0 - A_{\rm inf})_2 \right]  \varphi_0 \|^2  
+ \langle  \varphi_0 , ( B_{\omega_{-}} + V - E_{\rm inf} ) \varphi_0 \rangle . \nonumber
\end{eqnarray}
By Taylor expansion with remainder it is 
straight forward to see  that 
$$
( A_0(x) - A_{\rm inf}(x) )^2 \leq \frac{1}{4} 
\left[ \max_{|\alpha|=1} \| D^\alpha  B_{\omega_{-}} \|_\infty \right]^2 
\left\{ 4 | x_1 - \widetilde{x}_1 | | x_2 - \widetilde{x}_2 | + 
\frac{1}{2} | x_1 - \widetilde{x}_1 |^2 + \frac{1}{2} | x_2 - \widetilde{x}_2 |^2 \right\} .
$$
Using this estimate and evaluating a Gaussian integral we find, 
$$
\| \left[ ( A_0 - A_{\rm inf})_1 + i ( A_0 - A_{\rm inf})_2 \right]  \varphi_0 \|^2 \leq 
4 \left[ \max_{|\alpha|=1} \| D^\alpha  B_{\omega_{-}} \|_\infty \right]^2 
\frac{1}{b_0^2} \| \varphi_0 \|^2 , 
$$ 
where we used that  $ 0 <  b_0  \leq B_{\rm inf}^{(0)}$, which follows from \eqref{lowbound}.
Using a Taylor   expansion up to first respectively second order and  that $V + B_{\omega_-}$ attains 
 in $\widetilde{x}$ its minimum $E_{\rm inf}$   we find, similarly, 
\begin{eqnarray*}
\langle  \varphi_0 , ( B_{\omega_{-}} + V - E_{\rm inf} ) \varphi_0 \rangle  \leq 
\min\left\{
\max_{|\alpha|=2} \| D^\alpha  ( B_{\omega_{-}} + V) \|_\infty  \frac{2}{b_0} \ ,  \  \ 
\max_{|\alpha|=1} \| D^\alpha ( B_{\omega_{-}} + V) \|_\infty  \sqrt{\frac{2}{\pi b_0}  }
\right\}           \| \varphi_0 \|^2  .
\end{eqnarray*}
Now inserting the above estimates  into the right hand side of \eqref{eq:upbouninfspec21} and using 
Theorem  \ref{thm:ibspec} and the estimate 
$$
\| D^\alpha  B_{\rm ran}^{\omega_-} \|_\infty \leq \sum_{k=0}^\infty 2^{|\alpha| k} |m_-^{(k)} |
\| D^\alpha  U \|_\infty  \leq    \sum_{k=0}^\infty 2^{|\alpha| k} e^{ - \rho k} 
\| D^\alpha  U \|_\infty    , \quad . 
$$ 
we obtain \eqref{eq:bottomofspec}. Thus we have shown (a). 

(c) 
Now we estimate the interior of the spectrum. Let $\e  > 0$.
Then by Theorem  \ref{thm:ibspec} there exists an $\omega^\e$ in the support of the
probability measure and a normalized  $\varphi \in C_0^\infty$ such that
\be \label{sweepbound0}
 \langle \varphi , H(A_{\omega^\e}) \varphi \rangle \leq     \Sigma_{\rm inf} + \e .
\ee
Choose $L_0'$ such that ${\rm supp} \varphi \subset \Lambda_{L_0'}$. Now choose  $L_0 \geq  L_0' + c_\delta$ in $\N$.
To show   \eqref{eq:estonspec}  we consider  the path
\be \label{eq:contpathomega}
(\omega_s)_z^{(k)} = {(\omega^{\e})}_{z'}^{(k)} + s ( m_+^{(k)} - {(\omega^{\e})}_{z'}^{(k)} ) , \quad  0 \leq s \leq 1 ,
\ee
where $z = n L_0 +  z'$ with $z' \in \Lambda_{L_0}$ and $n \in \Z^2$. Note that the configuration
 $\omega_s$ is $L_0$-periodic. 
We have
\be \label{sweepbound1}
\inf \sigma ( H(B_{\omega_0}) ) \leq  \Sigma_{\rm inf} + \e , \quad
  E_{\rm sup}  \leq \inf \sigma ( H (B_{\omega_1}) ) ,
\ee
where the first inequality follows from \eqref{sweepbound0} and \eqref{eq:contpathomega}, and 
the second inequality follows from  Lemma \ref{lem:estoninfspec}.

By perturbation theory it is known that for any $L>0$,  $\inf \sigma ( H_L(B_{\omega_s}))$ is 
a continuous function of
$s$. In Lemma  \ref{lem:unfiform} below we will show  the limit  of  
$ \inf \sigma ( H_{n L_0}(B_{\omega_s}))$ as $n \to \infty$ converges  to
$\inf \sigma ( H(B_{\omega_s}))$ uniformly in $s$. Thus
  $s \mapsto  \inf \sigma ( H(B(\omega_s))) $ is a continuous function of
$s \in [0,1]$. In view of  Theorem  \ref{thm:ibspec}
this continuity property  and  \eqref{sweepbound1} imply  the inclusion 
 \eqref{eq:estonspec}, since $\e > 0$  is   arbitrary.
\qed

\begin{lemma}\label{lem:unfiform} Suppose the assumptions of Theorem  \ref{thm:estonspec} hold and suppose
   $\omega_s \in \Omega$ is as defined in \eqref{eq:contpathomega}. Then there
exists a universal constant $C$   such that
 $$
| \inf \sigma ( H_L(B_{\omega_s}) ) - \inf \sigma ( H(B_{\omega_s})) | \leq \frac{C}{L^2} ,
$$
for all $s \in [0,1]$ and $L = n L_0$ with $n \in \N$.
\end{lemma}
\begin{proof} Set $B = B_{\omega_s}$ and $E_L(B) := \inf \sigma ( H_L(B) )$. 
For notational simplicity we drop the $\omega_s$
dependence, the estimate will be uniform in $\omega_s$.  By $L_0 \Z^2$-periodicity 
of the $B$ field, we have  for any $n \in \N$,
\be \label{eq:boundone}
\inf \sigma ( H(B) )  \leq E_L(B) .
\ee
To find a lower bound we use the I.M.S. localization formula,
$$
 H(A) = \sum_{z \in L \Z^2} J_z H(A) J_z -
\sum_{z \in L \Z^2} |\nabla J_z|^2    ,
$$
where we introduced a partition of unity $J_z = \varphi((x-z)/L)$,
with
$\varphi \in C_0^\infty(\R^2;[0,1])$, ${\rm supp} \varphi \in [-1,1]^2$,
 $\sum_{z \in \Z}\varphi^2(x-z) =1$,
and $C_\varphi := \| \sum_{z \in \Z} (\nabla \varphi)^2(x-z) \|_\infty < \infty $.
By the $L_0 \Z^2$-periodicity of the $B$ field, we find for any vector potential $A$ with $\nabla \times A = B$  
and any normalized $\psi \in C_0^\infty$,
$$
 \langle \psi , H(A) \psi \rangle  =  \sum_{z \in L \Z^2} \langle \psi , J_z H(A) J_z \psi
\rangle  -  \sum_{z \in L \Z^2} \langle \psi , |\nabla J_z|^2 \psi \rangle  
 \geq E_{2L}(B) \sum_{z \in L \Z^2} \| J_z \psi \|^2 - \frac{C_\varphi}{L^2}. 
$$
This implies
$$
\inf \sigma ( H(B) )\geq   E_{2L}(B) - \frac{C_\varphi}{L^2},
$$
which, together with \eqref{eq:boundone}, yields the lemma.
\end{proof}

\section{Initial length scale estimates} \label{sec:initiallength}

In this section we show an initial length scale estimate.
We define  $\widetilde{\Lambda} := \Lambda + [-c_\delta,c_\delta]^2$,
with $c_\delta$ as defined in  \eqref{def:cdelta}.

\begin{theorem}          \label{lem:probestoninfspec}
Assume that {\bf (A)}  holds and recall the definition of $\nu(\cdot)$ from \eqref{eq:probdist1}. 
 Then for  $h  > 0$
\begin{align} \label{eq:probestoninfspec}
{\bf P} \left\{  {\rm dist}( \inf \sigma (H_{\Lambda} ) ,  E_{\rm inf}  ) \geq  \mu h  \right\}
  \geq 1 -  | \widetilde{\Lambda} | \nu(  c_u^{-1} h ) .
\end{align}

\end{theorem}
\begin{proof} 
By Lemma  \ref{lem:estoninfspec},  $E_{\rm inf}$ is
a lower bound of the
infimum of the spectrum, thus
\begin{align}
{\rm l. h. s. \ \ of \ \ }  \eqref{eq:probestoninfspec}
&\geq {\bf P} \{  B_{\rm det}(x) + \mu B_{\rm ran}(x) + V(x)
\geq  \mu h   + E_{\rm inf}  , \quad \forall  x \in \Lambda \}
 \nonumber     \\
&\geq
 \left[ {\bf P} \{ \omega_0^{(0)} \geq  m_-^{(0)} +  c_u^{-1} h   \}
\right]^{ | \widetilde{\Lambda} | }   \nonumber \\
& \geq 1 -  | \widetilde{\Lambda} | \nu( c_u^{-1} h ) . \label{eq:fourthline}
\end{align}
The second line follows, since $\omega^{0}_z  \geq m_-^{(0)}  + c_u^{-1} h$ for all 
$z \in \widetilde{\Lambda}$ implies that for all $x \in \Lambda$
\begin{eqnarray*}
   B_{\rm det}(x) + \mu B_{\rm ran}(x) + V(x)           & \geq &
 B_{\rm det}(x)  + V(x)
 + \mu \sum_{k=0}^\infty \sum_{z \in \widetilde{\Lambda}^{(k)}} \omega_z^{(k)} u(x- z ) \\
& \geq &   B_{\rm det}(x)  + V(x)         + \mu  \sum_{z \in \widetilde{\Lambda}^{(0)}} c_u^{-1} h   u(x- z ) 
 + \mu \sum_{k=0}^\infty \sum_{z \in \widetilde{\Lambda}^{(k)}} m_-^{(k)}  u(x- z ) \\
& \geq &   B_{\rm det}(x)  + V(x)        + \mu  h    + \mu \sum_{k=0}^\infty
\sum_{z \in \widetilde{\Lambda}^{(k)}} m_-^{(k)}  u(x- z ) \\
&\geq &   E_{\rm inf}          + \mu  h   ,
\end{eqnarray*}
where we used the notation $\widetilde{\Lambda}^{(k)} =
{\Lambda}^{(k)} \cap \widetilde{\Lambda}$. 
Now \eqref{eq:fourthline} follows from the binomial formula.
\end{proof}

\begin{corollary}\label{lem:probestoninfspec2}   Assume that 
{\bf ($\boldsymbol{{\rm A}_\tau}$)} holds for some fixed
 $\tau > 2$ and   $c_v$.
 For any $\xi \in (0,  \tau - 2)$ set 
 $ \beta := \frac{1}{2}(1- \frac{\xi+2}{\tau}) \in (0,1)$, then
there is  an $l_{\rm initial} = l_{\rm initial}(\tau,  \xi,c_u,c_v,c_\delta)$ such that 
\begin{align*}
{\bf P} \left\{  {\rm dist}( \inf \sigma (H_{\Lambda} ) ,  E_{\rm inf}  )
\geq  \mu l^{\beta-1} \right\}    \geq 1 -  l^{-\xi} ,
\end{align*}
for any $\Lambda = \Lambda_l(x)$, with $x \in \Z^2$ and $l \geq l_{\rm initial}$.
\end{corollary}
\begin{proof}
 Set
$h = l^{ \beta - 1 }$ in Theorem
\ref{lem:probestoninfspec}. Then
$$| \widetilde{\Lambda} | \nu( c_u^{-1} h ) \leq  |\widetilde{\Lambda}| c_v (c_u^{-1} h)^\tau
 =  c_u^{-\tau}  c_v (l+ c_\delta)^{ 2} l^{    (\beta - 1 )  \tau } \leq l^{-\xi}, $$
where the first inequality follows from assumption {\bf ($\boldsymbol{{\rm A}_\tau}$)},
and the  second inequality holds  for large $l$.
\end{proof}

\section{Multiscale analysis} \label{sec:multiscale}

The goal of this section is to prove Theorem  \ref{thm:loclization}. We will essentially
follow the setup presented in \cite{S} and indicate the necessary modifications
for magnetic fields. Alternatively, one could follow the setup of \cite{CH}
and verify  their key  hypothesis $[H1](\gamma_0,l_0)$.

We assume  {\bf ($\boldsymbol{{\rm A}_\tau}$)} throughout this section  for some fixed
 $\tau > 2$ and   $c_v$.
The constants $b_0,  \rho, \delta$ are as in the assumptions of Theorem \ref{thm:wegner}.
We write
$$ R_\Lambda(z) = R_\Lambda(A,z) = (H_{\Lambda}(A) - z )^{-1} =  
 (H_{\Lambda}(A_\omega) - z )^{-1}  .
$$
For notational simplicity we will occasionally drop the $A$ and $z$, and mostly the
$\omega$ dependence.
Boxes with sidelength $l \in 2 \N + 1$ and center $x \in \Z^2$
are called {\it suitable}.
For a suitable square  $\Lambda = \Lambda_l(x)$, we set
$$
\Lambda^{{\rm int}} := \Lambda_{l/3}(x) , \quad 
\Lambda^{{\rm out}} := \Lambda_{l}(x) \setminus \Lambda_{l-2}(x) ,
$$
and we set  $\chi^{\rm int} = \chi_{\Lambda^{\rm int}}$ and $\chi^{\rm out} =
 \chi_{\Lambda^{\rm out}}$. For $A$ an operator in a Hilbert space we will denote 
by $\rho(A)$ the resolvent set of $A$.

\begin{definition}
A square $\Lambda$ is called {\bf $(\gamma,\Lambda)$-good} for $\omega \in \Omega$ if
$$
\| \chi^{\rm out} R_\Lambda(B_\omega, E) \chi^{\rm int}  \| \leq \exp(-\gamma l ) ,
$$
where  $E \in \rho(H_\Lambda(B_\omega))$.
\end{definition}

Let us introduce the  multiscale induction hypotheses.
Below we denote by   $I \subset \R$   an interval and assume $l \in 2 \N + 1$.
First, 
for  $\gamma > 0$, and $\xi > 0$ we introduce  the following  hypothesis.

\medskip

\noindent
$G(I,l,\gamma,\xi)$:
\quad
$\forall x, y \in \Z^2$, $| x - y |_\infty \geq l+c_\delta $, the following estimate holds:
$$
{\bf P} \{ \forall E \in I  | \ \Lambda_l(x) \ {\rm or } \  \Lambda_l(y) 
\ {\rm is } \ (\gamma,E){\rm -good \ for} \ \omega \} \geq  1 - l^{-2 \xi} .
$$

\smallskip\noindent
Note that this definition  includes a security 
distance $c_\delta$, to ensure the independence
of  squares.

\begin{lemma} \label{lem:G} For any  $\xi \in (0, \tau -2)$ there is an $l_G = l_G(\tau,\xi,c_u,c_v,c_\delta)$ such that for all $ l \geq l_G$,  $G(I,l,\gamma,\xi)$
holds with $\gamma= l^{\beta-1}$, $I = E_{\rm inf}  + [0,  \frac{1}{2} \mu l^{\beta-1}]$,
and  $\beta =    \frac{1}{2}(1- \frac{\xi+2}{\tau})        \in (0,1)$.
\end{lemma}
\begin{proof} Consider $\omega$ such that
 \begin{align} \label{lem:eq:boundonbeta}
 {\rm dist}( \inf \sigma (H_{\Lambda}(\omega) ) ,  E_{\rm inf}  ) \geq \mu  l^{\beta-1} .
\end{align}
If $E \in I$, then ${\rm dist}( H_{\Lambda}(\omega) , E) \geq  \frac{\mu}{2} l^{\beta-1}$.
Thus by the resolvent decay estimate, see Theorem  \ref{thm:combthomas}, we find
$$
\| \chi^{\rm int} ( H_\Lambda(\omega)  - E)^{-1} \chi^{\rm out} \|
\leq  \frac{4}{\mu}   l^{1- \beta}  \exp(
 - ( \mu l^{\beta-1} / 4  )^{1/2}  l /4  ) ,
$$
for $l \geq 4$. Since by Corollary \ref{lem:probestoninfspec2}  the bound
\eqref{lem:eq:boundonbeta} holds with probability greater than $1 - l^{-\xi}$ for
 any large $l \geq l_{\rm initial}$, it
follows that for sufficiently large $l$, $G(I,l,\gamma,\xi)$ is valid for
$\gamma= l^{\beta-1}$.  
\end{proof}

\medskip

For $\Theta > 0$, and $q > 0$ we introduce  the following hypothesis.
\medskip

\noindent
$W(I,l,\Theta,q)$:
\quad
For all $E \in I$ and $\Lambda = \Lambda_l(x)$, $x \in \Z^2$, the following estimate holds:
$$
{\bf P} \{  {\rm dist}(\sigma(H_\Lambda(\omega)), E)  \leq \exp(-l^\Theta) \} \leq l^{- q} .
$$

\medskip

\begin{lemma} \label{lem:W} Suppose the assumptions of Theorem \ref{thm:wegner} hold. 
Let  $\Theta > 0$,  $q > 0$, and $0 < \kappa \leq 1$. Let  $I \subset \R$ be a finite interval with 
$\inf I \geq b_0/2$.
Then there exists a  constant
$L_0^* = L_0^*(I,\Theta,q,K_0,K_1,\delta,\mu,\kappa,\rho)$ 
such that 
$W(I,l,\Theta,q)$ holds for all $l \geq L_0^* b_0^{\kappa}$. 
\end{lemma}
\begin{proof}
Let $0 \leq \eta \leq 1$, and $\Lambda = \Lambda_l$. Then using Markov inequality and 
 Theorem \ref{thm:wegner} we have 
 \begin{eqnarray}
  {\bf P} \{  {\rm dist}(\sigma(H_\Lambda(A)), E)  \leq \eta/2  \} &=&
{\bf P}( {\rm Tr} \chi_{E,\eta}( H_\Lambda(A)  ) \geq 1 )  \nonumber \\
&\leq&  \E( {\rm Tr} \chi_{E,\eta}( H_\Lambda(A)  )  ) \nonumber  \\
&\leq& C_0  \eta \mu^{-2} l^{C_1 (\kappa^{-1} + \rho)} , \label{eq:boundonweg}
 \end{eqnarray}
for some constants $C_0$ and $C_1$ and $l$ sufficiently large. In fact,
by  Theorem \ref{thm:wegner} there exists an $L_0^*$ such that \eqref{eq:boundonweg} holds
  for all $l \geq L_0^* b_0^{\kappa}$.
Now we  choose $\eta = 2 \exp(- l^\Theta)$. Then by possibly choosing $L_0^*$ larger
the right hand side of \eqref{eq:boundonweg} is  bounded by
$l^{-q}$ for all $l \geq L_0^* b_0^{\kappa}$.
\end{proof}

\medskip

Thus we have shown that under certain conditions  the 
induction hypothesis of the multiscale analysis can be verified.
The following three technical lemmas will be needed for the multiscale analysis.
The have been verified for nonmagnetic random Schr\"odinger operators, see 
\cite{S}. Here we prove that they also hold for magnetic Schr\"odinger operators.

\begin{lemma} \label{lem:INDY} (INDY) $H_\Lambda(A_\omega)$ is measurable with respect
to $\omega \in \Omega$, the Hamiltonian  $H_{\Lambda_l(x)}(A_\omega)$ is stationary
in $x \in \Z^2$  in the sense of
\eqref{eq:stationarity},  and $|R_\Lambda(A_\omega,z)(x,y)|$ for $x,y \in \Lambda$ and
 $|R_{\Lambda'}(A_\omega,z)(x',y')|$ for $x',y' \in \Lambda'$
 are
independent for disjoint suitable squares
$\Lambda$ and $\Lambda'$ with ${\rm dist}(\Lambda,\Lambda')  \geq c_\delta$.
\end{lemma}
\begin{proof}
The measurability  follows from standard arguments see for example \cite{S}
Proposition 1.2.6 or see also \cite{CL}.
The stationarity is shown in Theorem \ref{thm:gaugtrafo2} (b). The independence 
follows from the independence of the magnetic fields when  restricted to  squares which 
are separated by a distance which is larger than $c_\delta$. 
\end{proof}

\begin{lemma} \label{lem:WEYL}  \label{lem:tracestimedynamical}
 Let $J \subset \R$ be a bounded interval. 
\begin{itemize}
 \item[(a)] (WEYL)  There is a constant $C = C(J,\| V \|_\infty)$ such that
$$
{\rm Tr}\left[ {\bf 1}_J(H_\Lambda(A)) \right] \leq C | \Lambda | \quad {\it for \ \ all \  \ } \omega \in \Omega
$$
 and every square  $\Lambda$.
\item[(b)]   $\chi_\Lambda {\bf 1}_{J}(H(A)) \chi_\Lambda$ is trace class and there exists a constant $C$ such that
for every  square $\Lambda$
\be \label{eq:boundontr}
{\rm Tr}\left[ \chi_\Lambda {\bf 1}_{J}(H(A))\right]  \leq C | \Lambda |  .
\ee
\end{itemize}
\end{lemma}
\begin{proof} Part (a) follows from an application of the  Lieb-Thirring inequality,
 see \eqref{LTineq}. 
For part (b),
by cyclicity of the trace $ {\rm Tr}\left[ \chi_\Lambda {\bf 1}_{J}(H(A) ) \right]
 =  {\rm Tr}\left[
 \chi_\Lambda {\bf 1}_{J}(H(A)) \chi_\Lambda \right]$.
By the spectral theorem
$$
0 \leq \chi_\Lambda {\bf 1}_J(H(A)) \chi_\Lambda \leq C_{t,J} \chi_\Lambda e^{ - 2 t H(A)} 
\chi_\Lambda .
$$
and by the diamagnetic inequality
$$
   \Tr \chi_\Lambda e^{ - 2 t H(A)} \chi_\Lambda \le e^{ 2 t \| V_-\|_\infty}
  \int_\Lambda e^{2t\Delta}(x,x)\rd x \le Ct^{-1}
e^{ 2 t \| V_-\|_\infty} 
 |\Lambda|,
$$
where $V_- := \min(0,V)$. 
Choosing $t=1$, we obtain  (b).
\end{proof}

\begin{lemma} \label{lem:GRI} (GRI) There is a $C_{\rm geom} = C_{\rm geom}(\|  V\|_\infty)$
 such that for $\Lambda, \Lambda'$ suitable
squares with $\Lambda \subset \Lambda'$, and 
$\Gamma_1 \subset \Lambda^{\rm int}$, $\Gamma_2 \subset \Lambda' \setminus \Lambda$, 
the following inequality holds for all $z \in \rho(H_\Lambda) \cap \rho(H_{\Lambda'})$,
$$
\| \chi_{\Gamma_2} R_{\Lambda'}(z) \chi_{\Gamma_1} \| \leq C_{\rm geom} (1 + |z|)
  \| \chi_{\Gamma_2} R_{\Lambda'}(z) \chi_\Lambda^{\rm out} \| \| \chi_\Lambda^{\rm out} R_\Lambda(z) \chi_{\Gamma_1} \| ,
$$
where the norms are operator norms.
\end{lemma}
\begin{proof}
 Let $\Lambda = \Lambda_l(x)$. Choose
$\phi \in C_c^\infty(\Lambda_{l-1/2}(x))$ which is 1 on $\Lambda_{l-1}(x)$.
Let $\Omega$ be the interior of $\Lambda^{\rm out}$. Then
${\rm dist}(\partial \Omega, {\rm supp} \nabla \phi ) \geq 1/4 =: d$. Moreover,
 $\phi$ can be chosen such that $\| \nabla \phi \|_\infty$
is bounded, independent of $\Lambda$.
Then we have
\begin{align*}
\| \chi_{\Gamma_2} (H_{\Lambda'} - z )^{-1} \chi_{\Gamma_1} \| &= \|
\chi_{\Gamma_2} [ \phi (H_{\Lambda'} - z )^{-1} - (H_\Lambda - z )^{-1} \phi ] \chi_{\Gamma_1} \|  \\
&= \| \chi_{\Gamma_2}  (H_{\Lambda} - z )^{-1} W(\phi)  (H_{\Lambda'} - z )^{-1}   \chi_{\Gamma_1}  \| \\
& \leq \underbrace{\| \chi_{\Gamma_2} R_\Lambda(z) (p-A) \cdot  (\nabla \phi ) R_{\Lambda'}(z) \chi_{\Gamma_1} \|}_{=:I} +
\underbrace{\| \chi_{\Gamma_2}  R_\Lambda(z)  ( \nabla \phi)  \cdot (p-A)   R_{\Lambda'}(z) \chi_{\Gamma_1} \|}_{=:II}  ,
\end{align*}
where in the second line we used  the geometric resolvent identity,
\be \label{eq:gri}
(H_\Lambda  - z )^{-1} \phi = \phi (H_{\Lambda'} - z )^{-1} + (H_{\Lambda} - z )^{-1}  W(\phi)
(H_{\Lambda'} - z )^{-1} ,
\ee
with  $W(\phi) := [ \phi , H_{\Lambda'} ]  = i \nabla \phi \cdot (p - A) +  (p - A)  \cdot i \nabla \phi$.
Now we estimate the first term on the right hand side.  Choose $\widetilde{\Omega}$ with
${\rm supp} \nabla \phi \subset \widetilde{\Omega} \subset \Omega $, and 
${\rm dist}(\partial \Omega , \partial \widetilde{\Omega}) \geq  d/2$.
We estimate
\begin{align*}
 I &= \| \chi_{\Gamma_2} R_\Lambda(z) (p-A)  \cdot
(\nabla \phi) \chi_{\widetilde{\Omega}}  R_{\Lambda'}(z) \chi_{\Gamma_1} \| \\
&\leq \|  \chi_{\widetilde{\Omega}}  (p-A)
 R_\Lambda(z) \chi_{\Gamma_2} \| \| \chi_\Omega
R_{\Lambda'}(z) \chi_{\Gamma_1} \| \| \nabla \phi\|_\infty .
\end{align*}
We now claim that the first term can be estimated by
\begin{eqnarray*}
 \|  \chi_{\widetilde{\Omega}}  (p-a)
R_\Lambda(z) \chi_{\Gamma_2} \| \leq C ( 1  + |z| + \| V  \|_\infty  )\|  \chi_{\Omega}   R_\Lambda(z) \chi_{\Gamma_2} \| .
\end{eqnarray*}
To see this we use Lemma  \ref{lem:estongrad} from Appendix B, with 
$u = (H_\Lambda - z)^{-1} \chi_{\Gamma_2} f$ and $g =\chi_{\Gamma_2} f$, for some $f \in L^2(\Omega)$,
and note that $\chi_{\Gamma_2} f = 0$ in $\Omega$. This yields the
desired bound on Term $I$. The second term, Term $II$,
can be estimated similarly.
\end{proof}



\begin{lemma} \label{lem:edi} Let $H(A)$ be a magnetic Schr\"odinger operator with $A \in C^1$
and $\nabla \cdot A =0$ 
such that  for 
$|\alpha|=1$ we have  $\sup_{x \in \R^2} | D^\alpha  A(x) |(1+|x|)^{-1} < \infty$  and 
$\| D^\alpha V \|_\infty < \infty$.
\begin{itemize}
\item[(a)] For spectrally almost every $E \in \sigma(H(A))$ there exists a polynomially 
bounded eigenfunction corresponding to $E$, i.e., ${\bf 1}_\Delta(H(A))=0$  where  $\Delta$ is the set of
 all energies in $\R$ for which there does not exist a polynomially bounded eigenfunction.
\item[(b)] For every bounded set $J \subset \R$ there exists a constant $C_J$ such that every
generalized eigenfunction $u$ of
$H(A)$ corresponding to $E \in J \setminus \sigma(H_\Lambda(A))$ satisfies
$$
{\rm ( EDI)} \quad \quad
\| \chi^{\rm int}_\Lambda u \| \leq C_J \| \chi_\Lambda^{\rm out} (H_\Lambda(A) - E)^{-1}
 \chi^{\rm int}_\Lambda \| \| \chi_\Lambda^{\rm out} u \| ,
$$
where $H_\Lambda(A)$ denotes the restriction of $H(A)$ to $L^2(\Lambda)$ with Dirichlet, 
Neumann or periodic boundary conditions.
\end{itemize}
\end{lemma}
\begin{proof}
(a) Follows from a generalization  of Theorem C.5.4 in \cite{sim82schsem} to magnetic Schr\"odinger 
operators. The proof given there  generalizes to magnetic Schr\"odinger operators by means 
of the diamagnetic inequality and the following modification. 
The $L^2$ growth  estimate stated in (ii) of     Theorem C.5.2  \cite{sim82schsem} 
can be shown as in that paper by means of the diamagnetic inequality.
To show that   (ii) of     Theorem C.5.2  \cite{sim82schsem}  implies (iii) of that same theorem
one has 
to use elliptic regularity instead of the Harnack inequality which was used in \cite{sim82schsem}.
(b) Follows with minor modifications as in the proof of Lemma 3.3.2 in \cite{S} and  Lemma \ref{lem:estongrad}.
\end{proof}

\medskip

\noindent
{\bf Proof of Theorem \ref{thm:loclization}.} 
 Fix  $\xi \in (0, \tau - 2)$ and let $\beta = \frac{1}{2}(1- \frac{\xi+2}{\tau})$. 
By  Lemma  \ref{lem:G} there exists  an
${l}_{G} = {l}_{G}(\tau,\xi,c_u,c_v,c_\delta)$ such that
$G(I_{l},l,\gamma_{l},\xi)$ holds  with $I_{l} := E_{\rm inf}  + [0, \frac{1}{2} \mu l^{\beta-1}]$
and $\gamma_{l} := l^{\beta-1}$ for all $l \geq {l}_{G}$.
Then choose
$0 < \Theta < \beta/2$ and $q > 2$ and $0 < \kappa < \min( (2 - 2 \beta)^{-1},1)$.  
By Lemma \ref{lem:W}  there exists an
${l}_{W}^* $ (depending on  $\Theta,  q, K_0, K_1, \delta, \mu, \kappa , \rho$)
such that  $W( I_{l}, l , \Theta, q  )$ is satisfied for 
$l \geq l_{W}^* b_0^{\kappa}$ and thus also for 
$l \geq l_0 := \max(l_{W}^* b_0^{\kappa},l_{G})$. 
Moreover, by Lemmas \ref{lem:WEYL},  \ref{lem:INDY}, and  \ref{lem:GRI}
we can now apply the multiscale analysis as outlined in
\cite{S} for the interval $J_0 := I_{l_0}$ (Specifically the assumptions  of  Theorem 3.2.2 and
Corollary 3.2.6 in \cite{S} are now verified). Note that  
the properties stated  in Lemma
\ref{lem:INDY} are weaker than the corresponding properties 
stated in \cite{S}, but one
readily verifies that they are sufficient for the multiscale 
analysis.
Namely, there is a minor modification necessary due to the 
security distance, which we introduced
in the definition of $G(I,l,\gamma,\xi)$. 
For a detailed discussion of the necessary changes, see for 
example \cite{kirstosto98}.

Fix $\omega \in \Omega$.
Having established the application of the multiscale analysis we can now show that $H(A_\omega)$ has
 pure point spectrum in $J_0$ using
the following standard argument.
By Lemma  \ref{lem:edi} (a) there is a set $\widetilde{J}_0 \subset J_0$ with the following properties:
(i) for every $E \in \widetilde{J}_0$ there is a polynomially bounded eigenfunction
 $u$ of $H(A_\omega)$ corresponding to $E$,
(ii) $J_0 \setminus \widetilde{J}_0$ is a set of measure zero for the spectral resolution
 of $E_{H(A_\omega)}$. 

Take a
generalized eigenfunction $u$ with energy $E\in J_0$.
 By Lemma  \ref{lem:edi} (b) it satisfies (EDI).
Thus by Proposition 3.3.1 in \cite{S} $u$ must be exponentially decaying. Thus $E$ is an eigenvalue.
Since the Hilbert space is separable, it follows that $\widetilde{J}_0$ must be countable.
Thus the restriction of the spectral measure to $J_0$ is supported on the  countable set 
 $\widetilde{J}_0$, and therefore it must
be purely discontinuous. Thus the spectrum of $H(A_\omega)$ in $J_0$ is pure point. Moreover, the eigenfunctions
are exponentially decaying.
 Dynamical localization, i.e. \eqref{eq:dynloc}, follows  from  an application of Theorem 3.4.1.
 in \cite{S}. A necessary
condition for the  application of  Theorem 3.4.1.
 in \cite{S} is that 
\be \label{eq:estonp}
p < \min(2 \xi, \frac{1}{4}(q-2)) .
\ee
If  $p < 2 ( \tau - 2)$,  we can choose 
$\xi$ and $q$, such that the  multiscale analysis
can be applied, i.e.,   $\xi < \tau - 2$ and  $q>2$,  and that \eqref{eq:estonp} holds.
 (Notice that different choices for $\xi$ and $q$, will
affect the right endpoint of  $J_0$. Hence the interval for which we are able prove dynamical 
localization might be smaller than the interval for which we can prove pure point spectrum.)
We thus proved that that the spectrum in $J_0 = [E_{\rm inf},E_{\rm inf} + e_0 ] $, with 
$e_0 := \frac{1}{2} \mu l_0^{\beta-1}$, is pure point.

It remains to show that 
$J_0$
contains indeed spectrum. 
For simplicity, we first  consider the case $K_2=0$ and $V=0$.  We know from  
Theorem   \ref{thm:estonspec}  (d)  that in that case 
$E_{\rm inf} = \Sigma_{\rm inf}$ and hence 
$J_0 =  [\Sigma_{\rm inf} ,  \Sigma_{\rm inf} + e_0 ] \subset \Sigma $.

Now let us assume the general case. By possibly choosing $l_0$ larger 
 we can assume by Theorem  \ref{thm:estonspec} (b) that $E_{\rm sup} \geq E_{\rm inf} + e_0$.
{F}rom  Theorem    \ref{thm:estonspec} (a) we know that 
\be \label{eq:bdepend-1}
 E_{\rm inf} \leq \Sigma_{\rm inf} \leq E_{\rm inf} + K(b_0) ,
\ee
with 
$$
 K(b_0) := 4 K_2^2b_0^{-2} +  \min(K_2 b_0^{-1/2} , K_3 b_0^{-1} )  .
$$
For  $b_0$ sufficiently large, we have one the one hand $l_0 = l_W^* b_0^{\kappa}$ and 
on the other 
\begin{equation} \label{eq:bdepend}
 K(b_0) \leq   \frac{1}{4}\mu (l_W^* b_0^{\kappa})^{\beta - 1} = \frac{1}{2} e_0 .
\end{equation}
In particular $\Sigma_{\rm inf} <  E_{\rm sup}$. Applying Theorem  \ref{thm:estonspec} (c)
we get
\be \label{eq:specsubset}
 [ \Sigma_{\rm inf}, E_{\rm sup} ] \subset \Sigma .
\ee
Now   \eqref{eq:specsubset}        and \eqref{eq:bdepend}  imply  that
$$
J_0 \cap \Sigma =  [\Sigma_{\rm inf}, E_{\rm inf} + e_0  ] = 
[\Sigma_{\rm inf}, \Sigma_{\rm inf} + {e}_1 ]
$$
for some ${e}_1 > 0$, see the figure below.

\vspace{1cm}

\noindent
\setlength{\unitlength}{5mm}
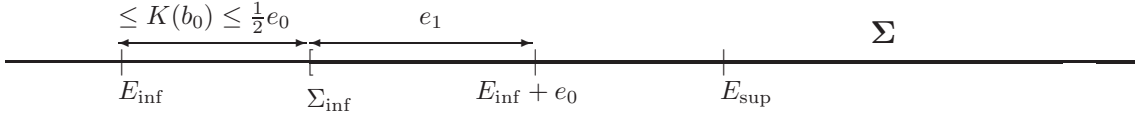
\begin{picture}(20,2)(-10,-1)

\put(-6,1){$\leq K(b_0) \leq \frac{1}{2}e_0$}
\put(-5,0.5){\vector(1,0){4}}
\put(-2,0.5){\vector(-1,0){4}}

\put(2,1){$e_1$}
\put(0,0.5){\vector(1,0){5}}
\put(4,0.5){\vector(-1,0){4.9}}

{ \put(14,0.5){\large{$\boldsymbol{\Sigma}$}}}

\put(-9,0){\line(1,0){30}}
\put(-6,-0.2){$|$}

  \put(-6,-1){$E_{\rm inf}$}

\put(-1,-0.2){[} \put(-1,-1.2){${\Sigma_{\rm inf}}$}
{
 \linethickness{0.5mm}
\put(-0.9,0){\line(1,0){20}}
\put(20,0){\line(1,0){1}}}
\put(3.5,-1){$E_{\rm inf} + e_0$}
\put(5,-0.2){$|$}
\put(10,-1){$E_{\rm sup} $}
\put(10,-0.2){$|$}
\end{picture}
\setlength{\unitlength}{5mm}

\qed

\appendix

\section{ Ergodicity}

{\bf Proof of Theorem  \ref{thm:deterministic}, Part 1}. 
The measurability of $H_\Lambda(A_\omega)$ for a 
finite box follows  from an easy application of  Proposition 1.2.6. in \cite{S}.
Let $f,g \in C_0^\infty$. For any $z$ with nonzero imaginary part we have 
\be \label{eq:measureab}
\lim_{l\to \infty} \langle f , (H_{\Lambda_l}(A_\omega) - z)^{-1} g \rangle  =  
\langle f , (H(A_\omega) - z)^{-1} g \rangle .
\ee
To this end we  can use the geometric resolvent equation \eqref{eq:gri}, 
and the resolvent decay estimate of Theorem \ref{thm:combthomas}. Since the limit
of measurable functions is measurable 
 \eqref{eq:measureab} implies the measurability of the magnetic Hamiltonian on $\R^2$
\qed

\medskip

\noindent
For $a \in \R^2$ we define the shift
operator $\mathcal{T}_a$ acting on functions $f$ on $\R^2$ by
$(\mathcal{T}_a f)(x) = f(x -a )$.
The operator $\mathcal{T}_a$ acts unitarily on  the Hilbert space $L^2(\R^2)$ 
and in that case we denote it by $\mathcal{U}_a$.
Given a  magnetic field $B : \R^2 \to \R $ we fix a gauge
for the vector potential   $A[B]: \R^2 \to \R^2$ by setting
$$
A[B](x_1,x_2) := \left( 0, \int_0^{x_1} B(x_1',x_2) dx_1' \right) .
$$
Note that
$
 \mathcal{T}_a B =  \nabla \times  (\mathcal{T}_a A[B])
$.
We define the function
$$
\lambda_a[B](x) := \int_{\gamma_{x}} \{ \mathcal{T}_a ( A[B]) - A[\mathcal{T}_a B] \} 
\rd s ,
$$
where  $\gamma_x$ is a path  in $\R^2$ connecting the origin with $x$
and $\rd s$ is the line integration measure.  
Since $\R^2$ is simply connected and
the rotation of the integrand is zero, the explicitly choice of $\gamma_x$ is not important.

{F}rom  the identity
$e^{ i \lambda_a[B]}  ( p - \mathcal{T}_a A[B] ) e^{- i \lambda_a[B]} =  p - 
A [\mathcal{T}_a B ]$ it follows that
\be
 e^{i \lambda_a[B]}  \mathcal{U}_a  H (A[B] ) \mathcal{U}_a^{*} e^{- i \lambda_a[B]} 
 =  H(A[ \mathcal{T}_a B ] ) .
\label{eq:gaugtrafo}
\ee
We define a family $(T_a)_{a \in \Z^2}$ of shift operators acting on $\Omega$ as 
$(T_a \omega)_z^{(m)} := \omega_{z - a}^{(m)}$.
As a trivial consequence of the definitions we have
\begin{equation} \label{eq:tranofranb}
B_{T_a \omega} = \mathcal{T}_a B_\omega  .
\end{equation}
\begin{proposition}\label{thm:gaugtrafo2} Let $a \in \Z^2$ and 
 $\mathcal{V}_a = e^{ i \lambda_a[B] } \mathcal{U}_a $.
Then the following holds.
\begin{itemize}
\item[(a)] We have
\begin{equation}   \label{eq:transform}
  \mathcal{V}_a  H(A[B_\omega] ) \mathcal{V}_a^{*}   =
H( A [ B_{\mathcal{T}_a \omega} ] ) ,
\end{equation}
i.e.,   $\omega \mapsto H( A[ B_{\omega} ] )$ is ergodic with respect to the family $(T_a)_{a \in \Z^2}$.
\item[(b)] For all  $\psi, \varphi \in C_0^\infty$,
\begin{equation} \label{eq:stationarity}
\langle  \mathcal{V}_a  \psi , H(A[B])  \mathcal{V}_a \phi \rangle \ \
 {\it  and} \ \    \langle  \psi , H(A[B]) \phi \rangle
\end{equation}
have the same probability distribution.
\end{itemize}
\end{proposition}
\begin{proof} (a) is a direct consequence of \eqref{eq:gaugtrafo} and  \eqref{eq:tranofranb}.
(b) From \eqref{eq:gaugtrafo}  it also follows that
$$
\langle  \psi , H(A[B]) \phi \rangle
=  \langle \mathcal{V}_a  \psi, H(A[\mathcal{T}_a B] )\mathcal{V}_a \phi \rangle .
$$
Now using \eqref{eq:tranofranb} and the  measure  preserving property of $\mathcal{T}_a$ part (b) follows.
\end{proof}

\medskip

\noindent
{\bf Proof of Theorem  \ref{thm:deterministic}, Part 2}. 
By  the ergodicity property  as stated in  Proposition
\ref{thm:gaugtrafo2} (a),
 Theorem \ref{thm:deterministic} can be obtained the same way as
Theorem~1.2.5 in \cite{S} using  the invariance of the trace
under conjugation by a unitary operator.
\qed

\section{Bound on the Magnetic Gradient}

We set  $\nabla_A = \nabla - i A$.   Let $z \in \C$.
We say that $u$ is a weak solution of $((p-A)^2  +V   )u = g$ in $\Omega$, if $u \in W^{1,2}(\Omega)$ and, for every
$\varphi \in C_c^\infty(\Omega)$,
\begin{equation} \label{eq:defofweaksol}
\langle \nabla_A \varphi,   \nabla_A u \rangle +   \langle      \varphi,  V  u  \rangle= \langle \varphi , g  \rangle .
\end{equation}
The following lemma is a minor modification of Lemma 2.5.3 in  \cite{S}.
\begin{lemma} \label{lem:estongrad}
 Let $\widetilde{\Omega} \subset \Omega \subset \R^2$ with 
${\rm dist}(\partial \Omega , \partial \widetilde{\Omega} ) =: d > 0$. Then
there exists a constant $C = C(d)$ such that every weak solution of
$H   u = g$ in $\Omega$ satisfies
$$
\| \nabla_A u \|_{L^2(\widetilde{\Omega})} \leq C (1 + \| V \|_\infty) 
\left( \| u \|_{L^2(\Omega)} + \| g \|_{L^2(\Omega)} \right) .
$$
\end{lemma}
\begin{proof}
Since $\overline{ C_c^\infty(\Omega)} = W_0^{1,2}(\Omega)$, Equation \eqref{eq:defofweaksol}
 holds for all $\varphi \in W_0^{1,2}(\Omega)$.
We can choose a function $\Psi \in C_c^\infty(\Omega)$, $0 \leq \Psi \leq  1$ with $\Psi \equiv 1$ on
$\widetilde{\Omega}$ and $\| \nabla \Psi \|_\infty \leq C d^{-1}$, where $C$ depends on 
the dimension $d$.
Set  $w:=  \Psi^2 u$. Then $w \in W_{0}^{1,2}(\Omega)$ and
$\nabla w = \Psi^2 \nabla u + 2 u \Psi \nabla \Psi$.
It follows from \eqref{eq:defofweaksol} that
$$
\langle \nabla_A w ,  \nabla_A u  \rangle +  \langle w ,  V  u \rangle = \langle w ,  g \rangle  .
$$
We obtain
\begin{eqnarray*}
 \| \Psi \nabla_A u \|^2 
&=& \langle  \nabla_A w , \nabla_A u  \rangle - 2 \langle u \nabla \Psi ,  \Psi \nabla_A u     \rangle \\
&=& \langle w, g \rangle +  \langle w,  V  u  \rangle
 - 2 \langle u \nabla \Psi , \Psi \nabla_A u   \rangle \\
&\leq&  \| g \| \| u \| +  \| V \|_\infty \| \Psi u \|^2 + 
\frac{1}{2}  \| \Psi \nabla_A  u \|^2 +  4 \| u \|^2  \| \nabla \Psi \|_\infty^2 .
\end{eqnarray*}
By the choice of $\Psi$ this now yields the claim.
\end{proof}

\section{Resolvent Decay Estimates}

Define the function  $\rho(x) = ( 1 + |x|^2)^{1/2}$. Let $\widetilde{H}$ be an operator of the form
$H_\Lambda(A)$. Define
\begin{equation}
\widetilde{H}(\alpha) := e^{ i \alpha \rho} \widetilde{H} e^{- i \alpha \rho}
= \widetilde{H} -  \alpha \nabla \rho \cdot ( - i \nabla - a) -   ( - i \nabla - a)  \cdot \alpha \nabla \rho 
  + \alpha^2 | \nabla \rho |^2 .
\label{eq:defofdecay}
\end{equation}
Since $|\nabla \rho|$ and $|\Delta \rho|$ are
bounded and $( - i \nabla - a)$ is infinitesimally small with respect to $\widetilde{H}$,
we obtain that $\widetilde{H}(\alpha)$ is an analytic family of type A on $\C$.

\begin{lemma} \label{lem:anab}
Let   $\beta \in \R$. Then  $( - \infty , \inf \sigma(\widetilde{H}) - \beta^2 )  \subset
\rho(\widetilde{H}(i \beta))$.
Let $z \in \C$ and  ${\rm Re} z   < \inf \sigma(\widetilde{H}) - \beta^2$, then
$$
\| ( \widetilde{H}(i \beta) - z  )^{-1} \| \leq
\frac{1}{  \inf \sigma(\widetilde{H}) - \beta^2 - {\rm Re} z   } .
$$
\end{lemma}
\begin{proof} Using $| \nabla \rho | \leq 1$, we find
 \begin{align*}
\| \psi \| \| ( \widetilde{H}( i \beta ) - z ) \psi  \|   & \geq  | \langle \psi , (
\widetilde{H}(i \beta) - z ) \psi \rangle  |  \geq
 | {\rm Re}   \langle \psi , ( \widetilde{H}(i \beta) - z  )  \psi \rangle  |  \\
&
 \geq  \langle  \psi  , ( \widetilde{H} - \beta^2 - {\rm Re} z     ) \psi \rangle
 \geq   ( \inf \sigma(\widetilde{H}) - \beta^2 -  {\rm Re} z     )   \| \psi \|^2 .
 \end{align*}
The lemma follows from this estimate.
\end{proof}

\begin{theorem} \label{thm:combthomas}  Let $\Lambda = \Lambda_l \subset \R^2$.
Let $E < {\rm inf} \sigma (H_\Lambda)$ and $\eta = {\rm dist}(E, {\rm inf} \sigma(H_\Lambda))$.
 Then, for $l\ge 4$,
$$
\| \chi^{\rm int} ( H_\Lambda - E)^{-1} \chi^{\rm out} \| \leq 
\frac{2}{\eta} \exp\Big( - \sqrt{\frac{\eta}{2}} \, \frac{l}{4}\Big).
$$
\end{theorem}
\begin{proof}
 Let $\varphi_1 , \varphi_2 \in C_c^\infty(\Lambda)$,
and $\alpha \in \R$. Then by unitarity
$$
I = \langle  \chi^{\rm int} \varphi_1 ,  ( H_\Lambda - E)^{-1} \chi^{\rm out}  \varphi_2 \rangle  =
\langle  e^{i \alpha \rho} \chi^{\rm int} \varphi_1 ,  
( H_\Lambda(\alpha)  - E)^{-1} \chi^{\rm out} e^{ i \alpha \rho}  \varphi_2  \rangle .
$$
By Lemma \ref{lem:anab}, we can analytically continue the
 resolvent occurring of the right hand side to a strip around the real axis  of width $\eta^{1/2}$.
Thus we find for $\alpha =  i \beta$ with  $\beta = \sqrt{\eta/2}$,
$$
I =   \langle   e^{  \beta  \rho} \chi^{\rm int} \varphi_1 , 
 ( H_\Lambda(i \beta )  - E)^{-1} \chi^{\rm out}
 e^{ -  \beta \rho}  \varphi_2 \rangle  .
$$
Using the resolvent estimate of Lemma  \ref{lem:anab} and inserting the definition of $\rho$, we find
$$
| I | \leq \| \varphi_1 \| \| \varphi_2 \| \frac{2}{\eta}
 \exp\Big( - \sqrt{\frac{\eta}{2}} \, \frac{l}{4}\Big).
$$
The theorem now follows.
\end{proof}

\medskip

\noindent
{\it Acknowledgement.}  The authors thank Peter M\"uller for
many stimulating discussions and insights on 
Wegner estimates.

\end{document}